\documentclass{article}

\usepackage[backend=bibtex]{biblatex}
\bibliography{refs}

\usepackage{eufrak}
\usepackage[affil-it]{authblk}
\usepackage{amsmath}
\usepackage{mathtools}
\usepackage{amssymb}
\usepackage{amsfonts}
\usepackage{makeidx}
\usepackage[dvipsnames]{xcolor}
\usepackage{tikz}
\usepackage{dsfont}
\usepackage{amsthm}
\usepackage{physics}
\usepackage{csquotes}
\usepackage{subcaption}
\usepackage[toc,page]{appendix}

\usepackage[a4paper, total={6in, 8in}]{geometry}

\usetikzlibrary{decorations.pathmorphing}
\tikzset{zigzag/.style={decorate, decoration=zigzag}}
\usetikzlibrary{cd}
\renewcommand{\restriction}{\mathord{\upharpoonright}}
\newcommand{\supp}{\textup{supp}}
\newcommand{\cat}{\mathbf}

\newtheorem{theorem}{Theorem}[section]
\newtheorem{lemma}[theorem]{Lemma}
\newtheorem{definition}[theorem]{Definition}
\newtheorem{corollary}[theorem]{Corollary}
\newtheorem{proposition}[theorem]{Proposition}

\begin{document}
	\title{Quantum fields on semi-globally hyperbolic space-times}
	\author{Daan Janssen%
		\thanks{Electronic address: \texttt{janssen@itp.uni-leipzig.de}}}
	\affil{Institute for Theoretical Physics,\\University of Leipzig,\\Germany}
	\maketitle
	
	\noindent We introduce a class of space-times modeling singular events such as evaporating black holes and topology changes, which we dub as semi-globally hyperbolic space-times. On these space-times we aim to study the existence of reasonable quantum field theories. We establish a notion of linear scalar quantum field theories on these space-times, show how such a theory might be constructed and introduce notions of global dynamics on these theories. Applying these constructions to both black hole evaporation and topology changing space-times, we find that existence of algebras can be relatively easily established, while the existence of reasonable states on these algebras remains an unsolved problem.
	
	\section{Introduction}
	In the past sixty years we could witness the emergence of a robust framework to describe quantum field theory on curved space-times. This framework, often referred to as local or algebraic quantum field theory, is versatile enough to study free quantum field theories, perturbative theories or even non-perturbative interacting theories, both in terms of specific integrable models and via an axiomatic approach, on a whole host of background space-times (see \cite{brunettiAdvancesAlgebraicQuantum2015} for an overview). However, not all background space-times lend themselves particularly well to constructing quantum field theories on them. In most cases where AQFT's are considered, the background space-times are assumed to be globally hyperbolic. This is an assumption on the global causal structure of a background that, not coincidentally, is often also imposed when studying the classical wave equation (or related PDE's) on curved space-times.

	While global hyperbolicity is a very convenient property in many contexts, specifically when one is interested in the global dynamics of a theory, every now and then one encounters a physically interesting space-time that is not globally hyperbolic. In some cases one would still like to construct a quantum field theory on these. Examples that may be of interest are space-times with closed time-like curves, i.e. space-times admitting `time-machines', space-times with time-like (conformal) boundaries, like AdS space-times, topology changing space-times, for instance space-times with dynamically formed wormholes, or space-times with more general naked singularities, such as a fully evaporating black hole space-time. While for some of these space-times it can be said that their physical relevance is mostly speculative, as is the case for wormhole space-times, one cannot ignore non globally hyperbolic space-times altogether.\\
	
	\noindent In recent years the space-times that have garnered the most attention when it comes to constructing quantum field theories on them, besides globally hyperbolic space-times that is, are those with time-like boundaries. Examples of work related to this are \cite{dappiaggiCasimirEffectPoint2015,wrochnaHolographicHadamardCondition2017,beniniAlgebraicQuantumField2018}. No doubt this is partly due to the advent of holographic dualities in physics, most notably the AdS/CFT correspondence which is closely associated with string theory, but also because some well-known classical solutions to the Einstein equations actually admit (singular) time-like boundaries, see for instance a Kerr black hole \cite{oneillGeometryKerrBlack1995}. One could argue that admitting time-like boundaries is actually only a very minor generalization with regards to global hyperbolicity. In fact many of the geometric properties of globally hyperbolic space-times carry over to `globally hyperbolic space-times with boundaries' without too much change \cite{hauStructureGloballyHyperbolic2019}. However on the side of field theories, both classical and quantum, such a generalization does warrant the introduction of appropriate boundary condition in order to fix time-evolution of the states in a theory.
	
	On the other side of the spectrum, in a sense very far removed from the class of globally hyperbolic space-times, are space-times with closed time-like curves. Quantum fields on  space-times such as the space-like cylinder have been considered in \cite{fewsterQuantumFieldTheory1996,kayQuantumFieldsCurved1997}, which lead to the introduction of a key property that many agree a quantum field on any (non globally hyperbolic) space-time should adhere to, namely that of F-locality. This property, very loosely speaking, entails that quantum field theories on any space-time should be locally equivalent to a theory on globally hyperbolic space-times. A more recent study of quantum fields on space-times with closed time-like curves can be found in \cite{tolksdorfQuantumPhysicsFields2018}. This work mostly focuses on investigating the D-CTC condition, a model for closed time-like curves in quantum computational network as introduced in \cite{deutschQuantumMechanicsClosed1991}, in the algebraic framework for quantum fields theory on curved space-times. However, this work also contains an explicit construction of a QFT on Pollitzer space-times, that indeed contain closed time-like curves and the resulting theory is in fact also F-local. While we do not consider space-times with closed time-like curves, F-locality will still be a key feature that we shall adhere to.\\
	
	\noindent So what space-times do we focus our attention on here and why? As mentioned, quantum fields on space-times with time-like boundaries have been relatively well studied (at least in the context of linear scalar fields), and shall therefore not be the focus of this work. Instead, we shall focus on space-times with more discrete (naked) singularities. A main motivation for studying quantum fields on such space-times comes from the supposed structure of evaporating black holes (originating in \cite{hawkingParticleCreationBlack1975}, and more recently studied from the perspective of causal structures in \cite{lesourdCausalStructureEvaporating2019}), however we shall in fact study a more general class of space-time for which black hole evaporation space-times are just one example, that we will treat in \ref{sec:bhe}. We shall dub this class the semi-globally hyperbolic space-times. Beyond the black hole evaporation space-times, or at least space-times of the conformal class that (semi-classically evolved) evaporating black holes are expected to be a part of, it also contains, for instance, the aforementioned topology changing space-times.
	
	As we will make precise in this paper, to us constructing a (linear scalar) quantum field theory specifically means constructing algebras generated by `operator valued distributions' on a space-time, satisfying an appropriate equation of motion. While in the case of globally hyperbolic space-times such an approach is in one-to-one correspondence with building an algebra from (symplectic smearings of) solutions to the classical equation of motion (see \cite{waldQuantumFieldTheory1994}), this correspondence need no longer hold if one treats more general space-times. This means that the problem of extending classical and quantum field theory to non-globally hyperbolic space-times can become somewhat further separated than on globally hyperbolic space-times.
	
	We rigorously define the class of semi-globally hyperbolic space-times in section \ref{sec:sgh}. In the section that follows, section \ref{sec:qft}, we propose a notion of an (f-local) linear scalar field theory on these space-times, describe how such a theory can be constructed and discuss how one should interpret notions of global dynamics on these theories. Thereafter, in section \ref{sec:maxsem} and \ref{sec:nonmax} we shall discuss applications of our construction to some example space-times, where we make the further distinction between so-called maximally semi-globally hyperbolic space-times, such as the black hole evaporation space-time, and other semi-globally hyperbolic space-times.
	
\section{Semi-globally hyperbolic space-times}
	\label{sec:sgh}
	In this section we define the our new class of space-times. However, before we write down this definition, let us recall how globally hyperbolic space-times are defined and what some of their important properties are. Many of these definitions and results, as well as a more general overview of global causal structures, can be found in \cite{minguzziCausalHierarchySpacetimes2008}.
	\begin{definition}
		\label{def:globhyp}
		A space-time $M$ (without boundaries) is \textup{globally hyperbolic} if there exists a (smooth) hypersurface $\Sigma\subset M$ such that each inextendible time-like curve in $M$ crosses $\Sigma$ exactly once. Such a surface $\Sigma$ is referred to as a (smooth) \textup{Cauchy surface}.\footnote{One can drop the smoothness condition without changing the class of space-times described by the definition, see \cite{bernalSmoothCauchyHypersurfaces2003}.}
	\end{definition}
	There is another way to characterize globally hyperbolic space-times which is often useful.
	\begin{proposition}
		A space-time $M$ is globally hyperbolic if and only if it is \textup{causal} (i.e. contains no closed causal curves) and for each $x,y\in M$ the set $J^+(x)\cap J^-(y)$ is compact.\footnote{Recall that $y\in J^+(x)$ if there is a future directed causal curve starting at $x$ and ending at $y$, and that $y\in J^-(x)$ if $x\in J^+(y)$. Furthermore, for a given set $U\subset M$ one defines $J^\pm(U)=\bigcup_{x\in U} J^\pm(x)$. The sets $I^\pm(x)$ and $I^\pm(U)$ are defined similarly, but with time-like curves instead of causal curves. }
	\end{proposition}
	Some important features of this class of space-times are as follows.
	\begin{proposition}
		Let $M$ globally hyperbolic with $\Sigma\subset M$ a smooth Cauchy surface, then
		\begin{enumerate}
			\item Given any $\Sigma'\subset M$ a further Cauchy surface, then $\Sigma\cong\Sigma'$, i.e. Cauchy surfaces are diffeomorphic.
			\item There exists a Cauchy time-function, i.e. a continuous function $T:M\rightarrow\mathbb{R}$ that strictly increases along each future directed causal curve such that each non-empty level set $T^{-1}(\{t\})$ is a Cauchy surface.
			\item Given a Cauchy time-function $T:M\rightarrow \mathbb{R}$ and $\Sigma_t=T^{-1}(\{t\})$ for some $t\in T(M)$, then $M\equiv T(M)\times\Sigma_t$.\footnote{This is known as Geroch's Theorem \cite{gerochDomainDependence1970}.} 
		\end{enumerate}
	\end{proposition}

	\noindent The existence of a global time-function on a space-time is a feature that goes under the name of \textit{stable causality}. This is one of a hierarchy of causality conditions, of which global hyperbolicity is the strongest. Clearly a stably causal space-time does not admit any closed time-like curves, but in fact this property is equivalent to that even when the metric of a stably causal space-time is perturbed (in a way made precise in \cite{minguzziCausalHierarchySpacetimes2008}), no closed time-like curves appear. This is a property we would like to keep in the class of space-times we are considering. Furthermore, as stated we are interested in space-times that, unlike in the case of time-like boundaries, only admit a discrete (i.e. locally finite) number of naked singularities. We can make this precise in the following way.
	\begin{definition}
		\label{def:sgh}
		Let $M$ a stably causal space-time. We say $M$ is \textup{semi-globally hyperbolic} if there exists a time-function $T:M\rightarrow \mathbb{R}$, where for each $a,b\in T(M)$ with $a<b$ there exist a finite set $\{t_i\}_{i=1}^n$ with $t_0=a$, $t_n=b$ and $t_i<t_{i+1}$ such that for each $i<n$ the (open) space-time $T^{-1}((t_i,t_{i+1}))$ is globally hyperbolic and $T\restriction_{T^{-1}((t_i,t_{i+1}))}$ is a Cauchy time-function. We call such a time-function \textup{semi-Cauchy}.
	\end{definition}
	\noindent Some examples from this class are given by the Penrose diagrams in figure \ref{fig:semglobhyp}.
	\begin{figure}[h]
		\centering
		\includegraphics{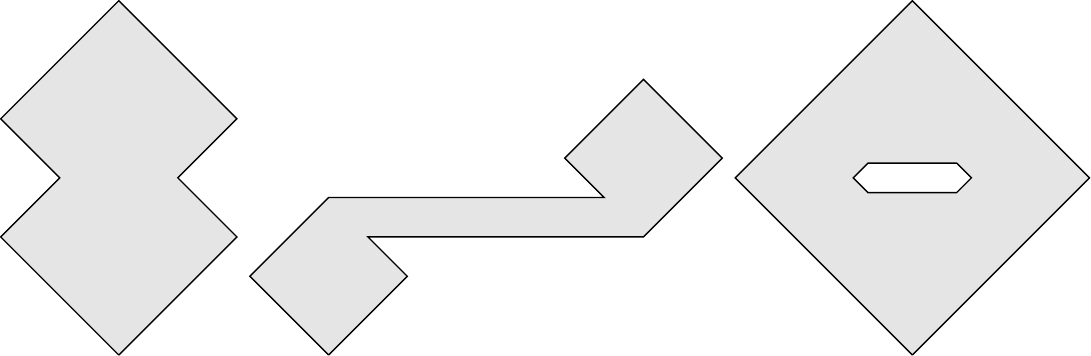}
		\caption{Some Penrose diagrams of two-dimensional semi-globally hyperbolic space-times}
		\label{fig:semglobhyp}
	\end{figure}
	It should be noted that many semi-globally hyperbolic may be embedded in to globally hyperbolic space-times, although generally not as a causally convex subspace-time. In fact the most simple examples of these are globally hyperbolic space-times from which a point has been removed. These and some similar space-times are discussed in more detail in section \ref{sec:maxsem}. We furthermore note that, just as globally hyperbolic space-times (see \cite{fewsterAlgebraicQuantumField2015}), the semi-globally hyperbolic space-times form a category where the morphisms are given by causally convex embeddings. To show this, we give an alternative description of semi-globally hyperbolic space-times, which will also prove useful in constructing quantum field theories.
	\begin{proposition}
		\label{thm:cover}
		A space-time is semi-globally hyperbolic iff for each connected component there is a countable cover $\{M_j:j\in J\}$ of open causally convex globally hyperbolic space-times such that there is a locally finite total order $\leq$ on $J$ satisfying for $j< j'$
		$$M_j\cap M_{j'}\neq\emptyset\text{ if and only if $j,j'$ nearest neighbours},$$
		$$I^+(M_j\cap M_{j'})\cap(M_j\cup M_{j'})\subset M_{j'}, \text{ and }I^-(M_j\cap M_{j'})\cap(M_j\cup M_{j'})\subset M_j.$$
	\end{proposition}
	\begin{proof}
		Without loss of generality assume $M$ is connected. Suppose $M$ is semi-globally hyperbolic and $T:M\rightarrow \mathbb{R}$ a semi-Cauchy time function. Note that for any $n\in\mathbb{Z}$ we can find a finite set $J_n=\{t_0,...,t_m\}\subset[n,n+1]$ such that $t_0=n$ $t_m=n+1$ and $t_i<t_{i+1}$ such that $T^{-1}((t_i,t_{i+1}))$ causally convex globally hyperbolic. Now define $J=\cup_{n\in\mathbb{Z}}J_n$, which is a countable locally finite totally ordered set $J=\{j_n:n\in\mathbb{Z}\}\subset\mathbb{R}$. For $j_n,j_{n+1}\in J$ nearest neighbours such that $j_n<j_{n+1}$, we set $$M_{j_n}=D(T^{-1}((j_n,j_{n+1})))\cap T^{-1}\left(\left(\frac{j_{n-1}+j_n}{2},\frac{j_{n+1}+j_{n+2}}{2}\right)\right).\footnote{Here $D(U)\subset M$ is the \textit{domain of dependence}, i.e. the set of all points $x\in M$ such that all inextendible causal curves through $x$ intersect $U\subset M$.}$$
		Note that $M_j$ is causally convex globally hyperbolic, as $D(T^{-1}((j_n,j_{n+1})))$ is causally convex globally hyperbolic and $T^{-1}\left(\left(\frac{j_{n-1}+j_n}{2},\frac{j_{n+1}+j_{n+2}}{2}\right)\right)$ causally convex. Furthermore $M_{j_n}\cap M_{j_m}\neq\emptyset$ if and only if $\vert n-m\vert\leq 1$.
		
		We now show that $\Sigma_{j_n}\subset M_{j_n}$. Suppose $x\in\Sigma_{j_n}= T^{-1}(\{j_n\})$, then clearly $$x\in T^{-1}\left(\left(\frac{j_{n-1}+j_n}{2},\frac{j_{n+1}+j_{n+2}}{2}\right)\right).$$ Now, let $\gamma_x:\mathbb{R}\rightarrow M$ be an inextendible causal curve through $x$. Since $T$ is a time function, and hence strictly increasing along causal curves, there must be a $x'\in \gamma_x(\mathbb{R})$ such that $j_n<T(x')<j_{n+1}$, which means $\gamma_x$ goes through $\Sigma_{T(x')}\subset M_{j_n}$. Since this holds for any time-like curve through $x$, we see $x\in M_{j_n}$ and hence $\Sigma_{j_n}\subset M_{j_n}$. Since this means that for each $t\in\mathbb{R}$ there is an $n\in\mathbb{Z}$ such that $T^{-1}(\{t\})\subset M_n$, we see that $$M=\bigcup_{j\in J}M_j.$$
		
		Observe that $U_n=M_{j_n}\cap M_{j_{n+1}}$ is itself causally convex and hence globally hyperbolic, that $\Sigma_{j_{n+1}}$ is a Cauchy surface for $U$ and that $I^+(U)=U\cup I^+(\Sigma_{j_{n+1}})$. Now suppose $$x\in I^+(\Sigma_{j_{n+1}})\cap (M_{j_n}\cup M_{j_{n+1}}),$$ let $T(x)>t'>j_{n+1}>t>j_n$, then any inextendible causal curve trough $x$ crosses either $\Sigma_{t'}$ (i.e. $x\in M_{j_{n+1}}$) or it crosses $\Sigma_t$ (i.e. $x\in M_{j_n}$), but in the latter case by continuity of $T$ we also see that the curve must cross $\Sigma_{t'}$, from which we conclude that $$I^+(\Sigma_{j_{n+1}})\cap (M_{j_n}\cup M_{j_{n+1}})\subset M_{j_{n+1}}.$$
		This implies 
		$$I^+(M_{j_n}\cap M_{j_{n+1}})\cap(M_{j_n}\cup M_{j_{n+1}})\subset M_{j_{n+1}},$$
		and a similar argument shows that 
		$$I^-(M_{j_n}\cap M_{j_{n+1}})\cap(M_{j_n}\cup M_{j_{n+1}})\subset M_{j_{n}}.$$
		This proves that every semi-globally hyperbolic space-time has a cover with the desired properties.\\
		\newline
		Now suppose we have a countable cover $\{M_j:j\in J\}$ as in the proposition. Without loss of generality we can assume that $J= \mathbb{Z}$. Observe that $M_n\cap M_{n+1}\neq\emptyset$ is globally hyperbolic and hence we can define for each $n$ a space-like Cauchy surface $\Sigma_n$ of $M_n\cap M_{n+1}$, which means that $M_n\cap M_{n+1}\subset I^-(\Sigma_n)\cup\Sigma_n\cup I^+(\Sigma_n)$. Now define 
		$$\tilde M_n=M_n\setminus (\Sigma_{n-1}\cup I^-(\Sigma_{n-1})\cup \Sigma_{n}\cup I^+(\Sigma_{n})).$$
		First observe that $\left(I^-(\Sigma_{n-1}) \cup\Sigma_{n-1}\right)\cap M_\subset M_n$ is a (relatively) closed subset on $M_n$, as 
		$$M_n\cap(\Sigma_{n-1}\cup I^-(\Sigma_{n-1}))=\left(M_n\cap M_{n-1}\right)\setminus I^+(\Sigma_{n-1}),$$
		and $I^+(\Sigma_{n-1})$ is open, which follows from \cite[cor. 2.9]{penroseTechniquesDifferentialTopology1972b}. Similarly $\left(I^+(\Sigma_{n}) \cup\Sigma_{n}\right)\cap M_\subset M_n$ is closed. This means in particular that
		$\tilde M_n$ is open.
		It is also causally convex, as for $a,b\in \tilde{M}_{n}$ we have $J^+(a)\cap J^-(b)\subset M_n$ compact, and if we suppose $J^+(a)\cap J^-(b)\cap \Sigma_{n-1}\cap I^-(\Sigma_{n-1})\neq\emptyset$, this means $a\in J^-(\Sigma_{n-1})\cup\Sigma_{n-1}\subset \overline{I^-(\Sigma_{n-1})}\cup\Sigma_{n-1}$, but as $\left(I^-(\Sigma_{n-1}) \cup\Sigma_{n-1}\right)\cap M_n$ is closed in $M_n$, we see that $a\in \Sigma_{n-1}\cup I^-(\Sigma_{n-1})$, which contradicts $a\in \tilde M_n$. A similar contradiction can be derived for $\Sigma_{n}\cup I^+(\Sigma_{n})$ from which we see that $J^+(a)\cap J^-(b)\subset \tilde M_n$ compact and hence $\tilde M_n$ is globally hyperbolic. Now we can define a Cauchy time-function $T_n:\tilde M_n\rightarrow(n-1,n)$.
		
		Let $\gamma:(0,1)\rightarrow \tilde M_n$ a future directed inextendible time-like curve such that $\lim_{\lambda\rightarrow 1}\gamma(\lambda)=\Sigma_n$. Clearly $T_n\circ\gamma$ is a bounded monotonically increasing function, so $\lim_{\lambda\rightarrow 1}(T_n\circ\gamma)(\lambda)\leq n$ exists. Also, since $\gamma$ is inextendible and for each $t\in (n-1,n)$ there is a Cauchy surface $\Sigma_t$ with $T(\Sigma_t)=\{t\}$, then for each $t\in (n-1,n)$ there will be an $\lambda\in (0,1)$ such that $(T\circ\gamma)(\lambda)=t$, hence 
		$$\lim_{\lambda\rightarrow 1}(T_n\circ\gamma)(\lambda)= n.$$
		A similar argument can be made for curves approaching $\Sigma_{n-1}$ and as a result we see that we can define a unique continuous extension $T_n:(\tilde M_n\cup \Sigma_{n-1}\cup\Sigma_n)\rightarrow [n-1,n]$ satisfying $T_n(\Sigma_n)=\{n\}$ and $T_n(\Sigma_{n-1})=\{n-1\}$.
		
		From our assumptions we clearly see that for $n\neq m$ we have $\Sigma_n\cap\Sigma_m=\emptyset$. Moreover, suppose $x\in\tilde M_n\cap\tilde M_m\subset M_n\cap M_m$, then $n,m$ nearest neighbours. Assume without loss of generality that $m=n+1$. It follows that $x\in \Sigma_n$, $x\in I^-(\Sigma_n)$ or $x\in I^+(\Sigma_n)$, where we note that these three options are mutually exclusive. Since $x\in \tilde M_n$, this means $x\not\in I^+(\Sigma_n)\cup\Sigma_n$, hence $x\in I^-(\Sigma_n)$, so $x\not\in \tilde M_{n+1}$, which is a contradiction. Hence we see that $\tilde M_n\cap\tilde M_m=\emptyset$. What follows is that 
		$$M=\bigcup_{n\in\mathbb{Z}}\tilde M_n\cup\Sigma_n,$$
		is a disjoint union and we can uniquely extend all $T_n$'s to a semi-Cauchy time-function $T:M\rightarrow \mathbb{R}$ such that $T\restriction_{\tilde M_n}=T_n$. This implies that $M$ is semi-globally hyperbolic.
	\end{proof}

	\noindent Since for each causally convex $U\subset M$ for $M=\bigcup_{j\in J}M_j$ connected and semi-globally hyperbolic with the cover satisfying the properties of proposition \ref{thm:cover}, the cover $U=\bigcup_{j\in J}(U\cap M_j)$ satisfies the same properties (for each connected component of $U$). The result below follows.
	\begin{corollary}
		Let $M$ semi-globally hyperbolic and $U\subset M$ causally convex. Then $U$ is also semi-globally hyperbolic.
	\end{corollary}
	\noindent As these aforementioned covers play a crucial role in the rest of this paper, we ought to give them a name.
	\begin{definition}
		Let $M$ semi-globally hyperbolic with a cover $\{M_j:j\in J\}$ satisfying the properties of proposition \ref{thm:cover}, we say this is a \textup{time-ordered cover}.
	\end{definition}
	\noindent Given the examples in figure \ref{fig:semglobhyp}, we can draw in (some choice of) a time-ordered cover as can be seen in figure \ref{fig:semglobhypcover}.
	\begin{figure}[h]
		\centering
		\includegraphics{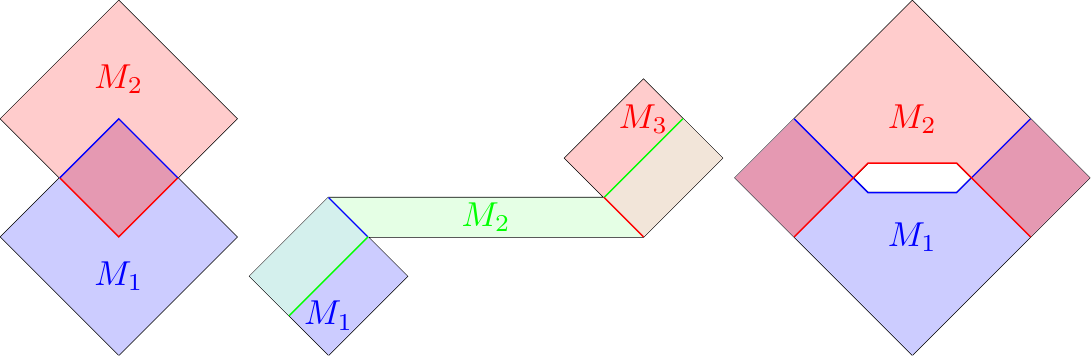}
		\caption{Some Penrose diagrams of two-dimensional semi-globally hyperbolic space-times with time-ordered covers drawn in}
		\label{fig:semglobhypcover}
	\end{figure}
	These time-ordered covers are in general not unique, but for a certain subclass of semi-globally hyperbolic space-times we can associate a preferred cover.
	\begin{definition}
		\label{def:max_cover}
		We say a space-time $M$ is \textup{maximally semi-globally hyperbolic} if it is semi-globally hyperbolic and there is a time-ordered cover $\{M_j:j\in J\}$ such that for each causally convex globally hyperbolic $U\subset M$ there is a $j\in J$ such that $U\subset D(M_j)$ and such that $M_j\subset D(M_{j'})$ implies $D(M_j)=D(M_{j'})$. We refer to the cover $\{D(M_j):j\in J\}$ as the \textup{maximal cover}.
	\end{definition}
	\begin{proposition}
		The maximal cover on a maximally semi-globally hyperbolic space-time is unique.
	\end{proposition}
	\begin{proof}
		Suppose $\{D(M_j):j\in J\}$ and $\{D(M_k):k\in K\}$ both maximal covers, then for any $j\in J$ the space-time $M_j$ is causally convex globally hyperbolic and hence there is a $k\in K$ such that $M_j\subset D(N_k)$, similarly there is a $j'\in J$ such that $N_k\subset D(M_{j'})$, which implies $D(M_{j})=D(M_{j'})$, from which is also follows that $D(N_k)=D(M_j)$. Hence for each $j\in J$ there is a $k\in K$ such that $D(M_j)=D(N_k)$ and vice versa.
	\end{proof}

	\noindent We have now characterized the class of space-times on which we shall attempt to define linear scalar quantum field theories. It is worth mentioning again that one subclass of these space-times that we will treat in section \ref{sec:bhe} are the (spherically symmetric) black hole evaporation space-times. These space-times are characterized by their (symmetry reduced) Penrose diagram, most famously drawn in \cite{hawkingParticleCreationBlack1975}. As can be seen in figure \ref{fig:bhe}, these space-times are in fact maximally semi-globally hyperbolic.
	\begin{figure}[h]
		\begin{center}
			\includegraphics{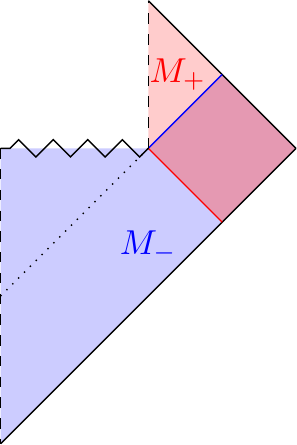}
		\end{center}
		\caption{Hypothesised Penrose diagram of a fully evaporating spherically symmetric black hole with its maximal cover}
		\label{fig:bhe}
	\end{figure}

	\section{Linear scalar quantum fields on semi-globally hyperbolic space-times}
	\label{sec:qft}
	For our purposes, (real) linear scalar quantum field theory is a quantization of a classical field satisfying the equation of motion
	\begin{equation}
		\label{eq:KG}
		(\Box-V)\phi=0.
	\end{equation}
	Here $\Box=\nabla^\mu\nabla_\mu$ with $\nabla$ the Levi-Civita connection on a given background space-time and $V$ some potential.\footnote{We take the metric to have $(-+++)$ signature.} If we want this potential $V$ to be locally covariant, it should take the form $V=m+\xi R$, with $m$ the mass of the field and $\xi$ a dimensionless coupling of the field to the Ricci scalar $R$.
	For the purposes of this paper we consider $m$ and $\xi$ to be fixed. The construction of the linear scalar quantum field theory on globally hyperbolic space-times, including some important features of this theory, are reviewed in appendix \ref{sec:qftapp}.

	\subsection{Algebras from states}
	\label{sec:constr}
	Based on the qualities of the real linear scalar quantum field on globally hyperbolic space-times, we set out some list of requirements that we want our quantum field theories on semi-globally hyperbolic space-times to satisfy. It should be noted that we do not assume that every semi-globally hyperbolic space-time $M$ admits a sensible quantum field theory. Luckily, we can take a page from \cite{kayQuantumFieldsCurved1997} to get us started on some very minimal requirements for a free scalar quantum field theory on any space-time. So for $M$ a non-globally hyperbolic space-time, we require that any real scalar quantum field theory (i.e. some unital *-algebra $\mathcal{A}(M)$) has the following properties.
	\begin{itemize}
		\item \textit{Net structure}: There is some set $\mathcal{O}$ of open neighbourhoods in $M$ containing at least all relatively compact subspaces, such that for each $U\in \mathcal{O}$ there is an algebra $$\mathcal{A}(M;U)\subset\mathcal{A}(M),$$ that in itself defines a quantum field theory on $U$, and where for any $U,V\in\mathcal{O}$ with $U\subset V$ we have $$\mathcal{A}(M;U)\subset\mathcal{A}(M;V).$$
		\item \textit{F-locality}: For each $x\in U\in\mathcal{O}$ there is a (a priori not necessarily causally convex) globally hyperbolic $N\in \mathcal{O}$ with $x\in N\subset U$ such that $$\mathcal{A}(M;N)\cong\mathcal{A}(N),$$
		via a net-preserving *-isomorphism, with $\mathcal{A}(N)$ the algebra of definition \ref{def:linscal}.
	\end{itemize}
	The net structure is essential to localize observables of a quantum field theory to regions in a space-time. F-locality on the other hand tells us that in some arbitrarily small globally hyperbolic neighbourhood our quantum field theory matches the standard theory on globally hyperbolic space-times (for some given potential $V$). This assumption can be understood from the fact that locally we are not able to discern the global causal structure of a space-time, so one would expect that at this local level one can also not distinguish between a quantum field theory on globally hyperbolic and non-globally hyperbolic space-times, which is fully in keeping with the traditions of local quantum field theory. That being said, one can argue that on the basis of this logic, the theory should match that of definition \ref{def:linscal} on any globally hyperbolic region, and not just on some arbitrarily small one. However in the context that F-locality was first discussed, namely in the presence of closed time-like curves, such a requirement is too restrictive to even begin to construct a quantum field theory \cite{kayPrincipleLocalityQuantum1992}. Nevertheless, on our semi-globally hyperbolic space-times we can often do a bit better, especially on maximally semi-globally hyperbolic space-times. To keep some further bounds on what kind of theories one would deem acceptable, we also introduce an additional requirement.
	\begin{itemize}
		\item The algebra $\mathcal{A}(M)$ can be generated by elements of the form $\hat\phi(f)$ for $f\in\mathcal{D}(M)$ such that for $f,g\in \mathcal{D}(M)$, the commutator $[\hat\phi(f),\hat\phi(g)]\in\mathcal{A}(M)$ is in the center of the algebra.
	\end{itemize}
	In principle one could allow theories with more involved commutation relations, but we deem this a natural generalization of the fact that in the globally hyperbolic case these commutators are always a multiple of the identity. At the very least this property requires that the commutation relation of our algebra is always specified to some degree, unlike for instance in the approach set out in \cite{beniniAlgebraicQuantumField2018}. As an upshot of this restriction, one can still sensibly define quasi-free states on such algebras.\\

	\noindent Let us capture all the required features in the following definition for linear scalar fields on semi-globally hyperbolic space-times.
	\begin{definition}
		\label{def:scalfield}
		Given a semi-globally hyperbolic space-time $M$ with time-ordered cover $\{M_i\subset M\}_{i\in I}$, we say a *-algebra $\mathcal{A}(M)$ is a \textup{scalar field algebra} on $M$ if there is a surjective *-homomorphism $\varphi:\mathcal{B}(M)\rightarrow\mathcal{A}(M)$,
		defining a net structure $$\mathcal{A}(M;U)=\varphi(\mathcal{B}(M;U))$$
		such that for each $i\in I$ there is a *-isomorphism $\iota_i:\mathcal{A}(M_i)\rightarrow\mathcal{A}(M;M_i)$, such that given the natural $\iota'_i:\mathcal{B}(M_i)\rightarrow\mathcal{B}(M)$ and surjective *-homomorphism $\varphi_i:\mathcal{B}(M_i)\rightarrow\mathcal{A}(M_i)$ (via the quotient of definition \ref{def:linscal}), the following diagram commutes
		\begin{center}
			\includegraphics{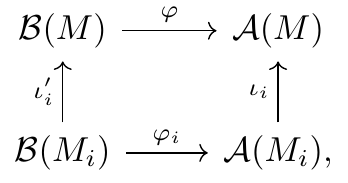}
		\end{center}
and where, denoting for $f\in\mathcal{D}(M)$ $$\hat\phi(f)=\varphi(\hat\psi(f)),$$
we require that for all $f,g,h\in\mathcal{D}(M)$
$$[[\hat\phi(f),\hat\phi(g)],\hat\phi(h)]=0.$$
\end{definition}

\noindent Note that the fact that commutators $[\hat\phi(f),\hat\phi(g)]$ are part of the center of the algebra also allow us to construct a Weyl-like C*-algebra $\mathfrak{A}(M)$ associated with the *-algebra $\mathcal{A}(M)$. Namely, as the group C*-algebra of $Z_{s.a.}(\mathcal{A}(M))\times\hat\phi(\mathcal{D}(M))$ with $Z_{s.a.}(\mathcal{A}(M))$ the self-adjoint elements of the center of the algebra and the group operation defined by $$\left(a,\hat\phi(f)\right)\cdot\left(b,\hat\phi(g)\right)=\left(a+b+\frac{i}{2}[\hat\phi(f),\hat\phi(g)],\hat\phi(f+g)\right).$$ If for $N\subset M$ one defines $\mathfrak{A}(M;N)\subset \mathfrak{A}(M)$ as the smallest C*-subalgebra containing all elements $\{0\}\times\hat\phi(\mathcal{D}(N))\subset\mathfrak{A}(M)$, then we see that for each $i\in I$ $$\mathfrak{A}(M;M_i)\cong\mathfrak{A}(M_i),$$
via a net-preserving isomorphism, with $\mathfrak{A}(M_i)$ the standard Weyl-algebra on the globally hyperbolic space-time $M_i$. We conclude that from a scalar field algebra $\mathcal{A}(M)$ we can construct an F-local net of C*-algebras on $M$.\\

\noindent As alluded to earlier, we have in fact imposed a stronger notion of F-locality on these quantum field theories than what is usually considered. One comment we should make here is that this definition depends on the choice of time-ordered cover (or via proposition \ref{thm:cover} a choice of semi-Cauchy time-function). It is a priori not guaranteed that a scalar field algebra with respect to some time-ordered cover will also be a scalar field algebra with respect to any other time-ordered cover. In the case of maximally semi-globally hyperbolic space-times this situation is less dire, as here one most naturally defines linear scalar algebras with respect to the maximal cover. This then guarantees that such an algebra is a linear scalar algebra with respect to any cover. Nevertheless, with our mind set on some applications in section \ref{sec:nonmax}, we don't want to restrict our definition to just maximally semi-globally hyperbolic space-times.

Having a *-algebra with a net structure is a good start, but of course one also is interested in what kind of states such an algebra admits. A particularly central question is whether a linear scalar algebra actually admits states in the first place. While the C*-algebra $\mathfrak{A}(M)$ associated with a scalar field algebra $\mathcal{A}(M)$ will always admit states (see \cite[Lemma 2.3.23]{bratteliAlgebrasNeumannAlgebras1987}), this does not mean that these states correspond to states on $\mathcal{A}(M)$, as $n$-point function of an arbitrary state on $\mathfrak{A}(M)$ are generally not defined. Since $n$-point functions are crucial objects in many uses of quantum field theory, especially those satisfying the Hadamard property (see \cite{fewsterArtState2018}), we view existence of $n$-point functions as imperative for any reasonable scalar field algebra. This requirement is the starting point of the construction in the rest of this section, where we particularly focus on 2-point functions, quasi-free states and one-particle structures, existence of which are of course equivalent to the existence of $n$-point functions for all $n$.

\subsubsection{A gluing procedure of one-particle structures}
As mentioned above, we shall be constructing quantum field theories on semi-globally hyperbolic space-times from (extended) one particle structures.
\begin{definition}
\label{def:ext_ops}
For a semi-globally hyperbolic space-time $M$ with time-ordered cover $\{M_j:j\in J\}$ we say $(K,H)$ with $H$ a Hilbert space and $K:\mathcal{D}(M)\rightarrow H$ a real linear map is an \textup{extended one-particle structure} if
$$\overline{K(\mathcal{D}(M))+iK(\mathcal{D}(M))}=H,$$
and for each $j\in J$ the pair $(K_j,H_j)$ with
$$K_j=K\restriction_{\mathcal{D}(M_j)},$$
and
$$H_j=\overline{K(\mathcal{D}(M_j))+iK(\mathcal{D}(M_j))},$$
form a one-particle structure on $(\mathcal{D}(M_j),\sigma_{M_j})$.
\end{definition}
\noindent Again here the choice of time-ordered cover influences what one would call an extended one-particle structure, but in any case we have the following.
\begin{proposition}
Given the assumptions of the definition above, an extended one-particle structure $(K,H)$ satisfies $$K\circ(\Box-V)=0.$$
\end{proposition}
\begin{proof}
Note that given an open cover $\{M_j:j\in J\}$ and $f\in\mathcal{D}(M)$, we can always find a finite $J'\subset J$ with for each $j\in J'$ an $f_j\in \mathcal{D}(M_j)$ such that $f=\sum_{j\in J'}f_j$. This means
$$K((\Box-V)f)=\sum_{j\in J'}K((\Box-V)f_j)=\sum_{j\in J'}K_j((\Box-V)f_j)=\sum_{j\in J'}0=0.$$
\end{proof}
\noindent For each extended one-particle structure one has a (real) pre-inner product $\mu$ and pre-symplectic structure $\sigma$ on $\mathcal{D}(M)$ via
$$\mu(f,g)=Re(\langle Kf,Kg\rangle),$$
$$\sigma(f,g)=2Im(\langle Kf,Kg\rangle).$$
Here we note that for any $j\in J$
$$\sigma\restriction_{\mathcal{D}(M_j)^2}=\sigma_{M_j},$$
with $\sigma_{M_j}$ the natural symplectic structure as defined in theorem \ref{thm:causprop}.\\
\newline
In general these one-particle structures will not be `god-given' and one has to construct them. One way to do this is via compatible locally defined one-particle structures.	The main result of this section is therefore as follows.
\begin{theorem}
A semi-globally hyperbolic space-time $M$ with time-ordered cover $\{M_j:j\in J\}$ admits an extended one-particle structure if and only if there is a chain of one-particle structures $\{(K_j,H_j):j\in J\}$, one for each globally hyperbolic patch of the cover, such that for each $j<j'$ nearest neighbours we have for $f,g\in\mathcal{D}(M_j\cup M_j')$ that
$$\langle K_jf,K_jg\rangle_j=\langle K_{j'}f,K_{j'}g\rangle_{j'}.$$
\end{theorem}
\noindent We prove this statement by explicitly constructing such an extended one-particle structure from a given chain via a procedure that we will refer to as a \textit{minimal gluing}.\\
\newline
To describe the minimal gluing of one-particle structure on patches of a general semi-globally hyperbolic space-time, we first treat the simpler case of a space-time that is just the union of two globally hyperbolic space-times.
\begin{definition}
\label{def:glued_ops}
Let $M$ a space-time and $M_1,M_2\subset M$ causally convex globally hyperbolic such that they form an open cover of $M$. Suppose for each $j\in\{1,2\}$ we have a one-particle structure $(K_j,H_j)$ for the pre-symplectic space $(D(M_j),\sigma_j)$ such that for $f,g\in\mathcal{D}(M_1\cap M_2)$
$$\langle K_1f,K_1g\rangle_1=\langle K_2f,K_2g\rangle_2,$$
we say these one-particle structures form a \textup{compatible pair}.

Now let $p_j$ the Hilbert space projection of $H_j$ onto $$O_j=\overline{K_j(\mathcal{D}(M_1\cap M_2))+iK_j(\mathcal{D}(M_1\cap M_2))},$$ and $U:O_1\rightarrow O_2$ the unitary such that $UK_1f=K_2f$ for each $f\in\mathcal{D}(M_1\cap M_2)$. We define the \textup{minimally glued one-particle structure} $(\tilde{K},\tilde{H})=(K_1,H_1)\Vert(K_2,H_2)$ with
$$\tilde H=H_1\oplus H_2/\{(\psi_1,\psi_2):\langle\psi_1,\psi_1\rangle_1+\langle\psi_2,\psi_2\rangle_2+\langle\psi_1,U^*p_2\psi_2\rangle_1+\langle\psi_2,Up_1\psi_1\rangle_2=0\},$$
where we define the inner product on $\langle.,.\rangle$ on $\tilde{H}$ as
$$\langle[(\psi_1,\psi_2)],[(\phi_1,\phi_2)]\rangle=\langle\psi_1,\phi_1\rangle_1+\langle\psi_2,\phi_2\rangle_2+\langle\psi_1,U^*p_2\phi_2\rangle_1+\langle\psi_2,Up_1\phi_1\rangle_2,$$
and with
$\tilde{K}:\mathcal{D}(M)\rightarrow \tilde{H}$ defined by $$\tilde K(f_1+f_2)=[(K_1f_1,K_2f_2)],$$
for $f_i\in\mathcal{D}(M_i)$.
\end{definition}
\noindent In the following lemma we prove that the definitions above are well-defined
\begin{lemma}
In the definition above $\langle.,.\rangle$ is a well-defined inner product on the Hilbert space $\tilde{H}$ and $\tilde{K}:\mathcal{D}(M)\rightarrow \tilde{H}$ is a well-defined real linear map.
\end{lemma}
\begin{proof}
Starting with the claims on the inner product, we first note that for $(\psi_1,\psi_2),(\phi_1,\phi_2)\in H_1\oplus H_2$ the map $$((\psi_1,\psi_2),(\phi_1,\phi_2))\mapsto \langle\psi_1,\phi_1\rangle_1+\langle\psi_2,\phi_2\rangle_2+\langle\psi_1,U^*p_2\phi_2\rangle_1+\langle\psi_2,Up_1\phi_1\rangle_2$$ is a positive semi-definite sesquilinear map (i.e. it is a pre-inner product). This follows immediately from the fact that we can rewrite
\begin{align*}
&\langle\psi_1,\phi_1\rangle_1+\langle\psi_2,\phi_2\rangle_2+\langle\psi_1,U^*p_2\phi_2\rangle_1+\langle\psi_2,Up_1\phi_1\rangle_2=\\&\langle p_1^\bot\psi_1,p_1^\bot\phi_1\rangle_1+\langle p_2^\bot\psi_2,p_2^\bot\phi_2\rangle_2+\langle Up_1\psi_1+p_2\psi_2,Up_1\phi_1+p_2\phi_2\rangle_2.
\end{align*}
We can now recognize that $\tilde{H}$ is the semi-norm reduction of $H_1\oplus H_2$ with respect to the pre-inner product above, hence $\langle.,.\rangle$ is well-defined.

For $\tilde{H}$ to form a Hilbert space with the inner product $\langle.,.\rangle$, we need to show it is complete. Suppose $[(\psi_{1,n},\psi_{2,n})]$ to form a Cauchy sequence in $\tilde{H}$. Using the rewritten form of the inner product, we can see that $p_1^\bot\psi_{1,n}$ is Cauchy in $p_1^\bot H_1$, $p_2^\bot\psi_{2,n}$ in $p_2^\bot H_2$ and $Up_1\psi_1+p_2\psi_2$ in $O_2$. Therefore we can define $\psi_1=\lim_{n\rightarrow\infty}p_1^\bot\psi_{1,n}\in p_1^\bot H_1$, $\psi_2=\lim_{n\rightarrow\infty}p_2^\bot\psi_{2,n}\in p_2^\bot H_2$ and $\psi_O=\lim_{n\rightarrow\infty}Up_1\psi_1+p_2\psi_2\in O_2$, and define $\Psi=[(\psi_1,\psi_O+\psi_2)]$. Using that $p_i\psi_i=0$, we now see that
\begin{align*}
\lim_{n\rightarrow \infty}&\Vert[(\psi_{1,n},\psi_{2,n})]-[(\psi_1,\psi_O+\psi_2)]\Vert=\\\lim_{n\rightarrow \infty}&\Vert p_1^\bot(\psi_{1,n}-\psi_1)\Vert_1+\Vert p_2^\bot(\psi_{2,n}-\psi_O-\psi_2)\Vert_2\\&+\Vert Up_1(\psi_{1,n}-\psi_1)+p_2(\psi_{2,n}-\psi_O-\psi_2)\Vert_2=\\
\lim_{n\rightarrow \infty}&\Vert p_1^\bot\psi_{1,n}-\psi_1\Vert_1+\Vert p_2^\bot\psi_{2,n}-\psi_2\Vert_2+\Vert Up_1\psi_{1,n}+p_2\psi_{2,n}-\psi_O\Vert_2=0,
\end{align*}
from which we conclude that $\tilde{H}$ is complete.

Now to show that $\tilde{K}$ is well defined, first observe that for any $f\in\mathcal{D}(M)$ there exist $f_i\in\mathcal{D}(M_i)$ such that $f=f_1+f_2$, as can for instance be seen using partitions of unity. Now suppose $f_1+f_2=g_1+g_2$ for $g_i\in\mathcal{D}(M_i)$, observe that $f_1-g_1=g_2-f_2\in\mathcal{D}(M_1\cup M_2)$.
We now see that
\begin{align*}
&\Vert[(K_1f_1,K_2f_2)]-[(K_1g_1,K_2g_2)]\Vert=\Vert[(K_1(f_1-g_1),K_2(f_2-g_2))]\Vert=\\
&\Vert UK_1(f_1-g_1)+K_2(f_2-g_2)\Vert_2=\Vert K_2(f_1-g_1+f_2-g_2)\Vert_2=\Vert K_20\Vert_2=0.
\end{align*}
Hence, $\tilde{K}$ is well defined.
\end{proof}

\noindent Using very similar arguments as above, we get the following.
\begin{corollary}
$\tilde K(D(M))+i\tilde K(D(M))$ is dense in $\tilde{H}$.
\end{corollary}
\noindent Since the maps $\tilde{K}\restriction_{\mathcal{D}(M_i)}$ define one-particle structures equivalent to $K_i$, $(\tilde K,\tilde H)$ indeed define extended one-particle structures. However generally this extension is not unique.\footnote{Uniqueness is guaranteed if one of the two one-particle structures, or rather their associated quasi-free states, satisfy the Reeh-Schlieder property (see for instance \cite{strohmaierMicrolocalAnalysisQuantum2002}), or at least satisfy it with respect to $M_1\cap M_2$. In fact suppose $(K_1,H_1)$ satisfies the Reeh-Schlieder property, this means that $p_1^\bot H_1=\{0\}$.}\\
\newline
Until now we have only considered space-times that were built from two globally hyperbolic patches. Let us now generalize our gluing procedure to arbitrary semi-globally hyperbolic space-times.
\begin{definition}
\label{def:gluedchain}
Let $M=\bigcup_{j\in J}M_j$ be semi-globally hyperbolic with time-ordered cover. Let $(K_j,H_j)$ be one-particle structures that are pairwise compatible, we then refer to $\{((K_j,H_j)):j\in J\}$ as a \textup{compatible chain} of one-particle structures. We now define $(\tilde K,\tilde H)=\Vert\{((K_j,H_j)):j\in J\}$ as follows.
Let $j,j'\in J$ with $j<j'$, for the interval $[j,j']$ we define $(\tilde K_{[j,j']},\tilde H_{[j,j']})$ inductively via
$$(\tilde K_{[j,j']},\tilde H_{[j,j']})=(\tilde K_{[j,j']\setminus\{j'\}},\tilde H_{[j,j']\setminus\{j'\}})\Vert(K_{j'},H_{j'}).$$
Now for $[j,j']\subset[k,k']\subset J$ we define
$\varphi_{[j,j'],[k,k']}:\tilde H_{[j,j']}\rightarrow \tilde H_{[k,k']}$ as the unique isometric embedding such that for $f\in\mathcal{D}\left(\bigcup_{l\in[j,j']}M_l\right)$ we have $$\varphi_{[j,j'],[k,k']}\tilde K_{[j,j']}f= \tilde K_{[k,k']}f.$$
We now define $\tilde H$ as the direct limit of the directed set $\{\tilde H_{[j,j']};j<j'\}$, i.e.
$$\tilde H=\left.\bigsqcup_{j,j'\in J,j<j'}\tilde H_{[j,j']}\middle/\sim\right.$$
where
$(H_{[j,j']},\psi_{[j,j']})\sim(H_{[k,k']},\psi_{[k,k']})$ if there are $l,l'\in J$ such that
$[j,j']\cup[k,k']\subset[l,l']$ and $\varphi_{[j,j'],[l,l']}\psi_{[j,j']}=\varphi_{[k,k'],[l,l']}\psi_{[k,k']}$.
This also yields an natural isometric embedding $\varphi_{[j,j'],J}:\tilde H_{[j,j']}\rightarrow \tilde H$.

For $f\in \mathcal{D}(M)$, take some $[j,j']\subset J$ such that $\text{supp}(f)\in\bigcup_{l\in[j,j']}M_l$, and now define
$$\tilde Kf= \varphi_{[j,j'],J}\tilde K_{[j,j']}f.$$
\end{definition}
\noindent Applying this definition to a slightly more simple setting, we get the following result from straightforward calculation.
\begin{proposition}
Let $M$ a semi-globally hyperbolic space-time $M$ with time-ordered cover $M=\bigcup_{j=0}^nM_i$ and a minimal extension $(\tilde K,\tilde H)=\Vert_{j=0}^n(K_i,H_i)$ defined from a compatible chain of one-particle structures as in definition \ref{def:gluedchain}, define $$p_j:H_j\rightarrow\overline{K_j(\mathcal{D}(M_j\cap M_{j-1}))+iK_j(\mathcal{D}(M_j\cap M_{j-1}))},$$ as the projector w.r.t the inner product $\langle.,.\rangle_j$ and 
\begin{align*}U_j:&\overline{K_j(\mathcal{D}(M_j\cap M_{j-1}))+iK_j(\mathcal{D}(M_j\cap M_{j-1}))}\\&\rightarrow \overline{K_{j-1}(\mathcal{D}(M_j\cap M_{j-1}))+iK_{j-1}(\mathcal{D}(M_j\cap M_{j-1}))}
\end{align*}
the natural unitary map, we find for $f\in\mathcal{D}(M_0)$ and $g\in\mathcal{D}(M_n)$
$$\langle \tilde Kf,\tilde Kg\rangle=\langle K_0f,U_1p_1...U_np_nK_ng\rangle_0.$$
Similarly, defining
$p_j':H_{j-1}\rightarrow\overline{K_{j-1}(\mathcal{D}(M_j\cap M_{j-1}))+iK_{j-1}(\mathcal{D}(M_j\cap M_{j-1}))}$ as the projector w.r.t. $\langle.,.\rangle_{j-1}$, we find
$$\langle \tilde Kf,\tilde Kg\rangle=\langle U_n^*p_n'...U_1^*p_1'K_0f,K_ng\rangle_n.$$
\end{proposition}
\noindent Here we see explicitly that the fact that our covers admit some natural time-ordering matters in the construction, as projectors $p_i$ and $p_j$, even when extended to $\tilde H$, generally do not commute, hence we need some (geometric) input to decide on the ordering.

\subsubsection{Constructing the algebra from one-particle structures}
Our goal of this section is as follows. Given a set of extended one-particle structures on a semi-globally hyperbolic space-time $M=\bigcup_{j\in J}M_j$, we want to construct the smallest scalar field algebra $\mathcal{A}(M)$ such that a given set of one-particle structures extend to Fock-space representations of this algebra.\\
\newline
For a single extended one particle structure, we define the associated scalar field algebra in the following way.
\begin{definition}
\label{def:1stextended_algebra}
Let $M$ semi-globally hyperbolic and $\{M_j:j\in J\}$ the ordered cover. Given $(K,H)$ an extended one-particle structure, we define the scalar field algebra $\mathcal{A}_{(K,H)}(M)$ as follows. Let $\mathcal{B}(M)$ the the pre-field algebra of linear observables on $M$. Now let $\mathcal{J}_{(K,H)}\subset \mathcal{B}(M)$ the smallest ideal such that for each $f,g\in\mathcal{D}(M)$ we have
$$[\hat\psi(f),\hat\psi(g)]-i\sigma(f,g)\in \mathcal{J}_{(K,H)},$$
and for each $f\in\mathcal{D}(M)$ satisfying
$$\mu(f,f)=0,$$
we have $\psi(f)\in \mathcal{J}_{(K,H)}$. We define
$$\mathcal{A}_{(K,H)}(M)=\left.\mathcal{B}(M)\middle/\mathcal{J}_{(K,H)}\right..$$
\end{definition}
\noindent That this algebra satisfies definition \ref{def:scalfield}, can be seen from the following lemma.
\begin{lemma}
\label{lem:ext_alg_emb}
Let $M$, $\{M_j:j\in J\}$, $(K,H)$ and $\mathcal{A}_{(K,H)}(M)$ as in definition \ref{def:1stextended_algebra}. For $j\in J$, let $\mathcal{J}_j\subset\mathcal{B}(M_j)$ the ideal of definition \ref{def:linscal} and $\iota'_j:\mathcal{B}(M_j)\rightarrow\mathcal{B}(M;M_j)$ the natural isomorphism. Then
$$\mathcal{J}_{(K,H)}\cap\mathcal{B}(M;M_j)=\iota'_j(\mathcal{J}_j).$$
\end{lemma}
\begin{proof}
Suppose $$a\in\mathcal{J}_{(K,H)}\cap\mathcal{B}(M;M_j)\subset\mathcal{B}(M;M_j),$$ this is equivalent to that $a=\sum_{n=0}^Na_n$ where for each $n$ there are $b_n,c_n\in\mathcal{B}(M;M_j)$ such that either
\begin{enumerate}
\item $f,g\in\mathcal{D}(M)$ with $\text{supp}(f),\text{supp}(g)\subset M_j$ such that
$$a_n=b_n([\hat\psi(f),\hat\psi(g)]-i\sigma(f,g))c_n$$
or
\item $f\in\mathcal{D}(M)$ with $\text{supp}(f)\subset M_j$ and $\mu(f,f)=0$ such that
$$a_n=b_n\hat\psi(f)c_n.$$
\end{enumerate}
In the first case, we clearly see that $a_n\in\iota'_j(\mathcal{J}_j)$, as we know that $\text{supp}(f),\text{supp}(g)\subset M_j$ $$\sigma(f,g)=\sigma_{M_j}(f,g).$$
In the second case, observe that
$$\vert\sigma_{M_j}(f,g)\vert^2\leq 4\mu(f,f)\mu(g,g).$$
It follows that
$\text{supp}(f)\subset M_j$ and $\tilde{\mu}(f,f)=0$ imply that $\sigma_{M_j}(f,g)=0$ for all $g\in\mathcal{D}(M;M_j)$. This implication also goes the other way, as we know that if $\sigma_{M_j}(f,g)=0$ for all $g\in\mathcal{D}(M;M_j)$, then $f\in(\Box-V)\mathcal{D}(M;M_j)$ and hence $\mu(f,f)=0$.

We conclude that $a\in\mathcal{J}_{(K,H)}\cap\mathcal{B}(M;M_j)$ if and only if $a=\sum_na_n$ with $a_n\in\iota'_j(\mathcal{J}_j)$, which is in turn equivalent to $a\in\iota'_j(\mathcal{J}_j)$.
\end{proof}

\noindent It should be clear that the quantum field algebra $\mathcal{A}_{(K,H)}(M)$ admits a quasi-free state $\omega$ for which the two-point function is given by $\mu$ via $$\omega\left(\hat\phi_{(K,H)}(f)\hat\phi_{(K,H)}(g)\right)=\langle K f, K g\rangle=\mu(f,g)+\frac{i}{2}\sigma(f,g).$$ In fact we can easily lift the discussion following definition \ref{def:fock} to this algebra to see that the GNS representation associated with the state above also yields a faithful Fock-space representation where the one-particle Hilbert space matches $H$.\\
\newline
We have seen how, assuming the fact that a semi-globally hyperbolic space-time admits a compatible chain of one-particle structures, we can construct a quantum field theory on this space-time that admits a Fock-space representation with a one-particle Hilbert space that corresponds to the gluing of this particular compatible chain of one-particle structures. However in general there may be many compatible chains that can be extended in several ways and one may not want to select a preferred one. We will reflect on this further when we treat some example space-times. For now suppose that we have some set $E$ consisting of numerous extended one-particle structures, we wish to construct an algebra that for each $(K,H)\in E$ admits a quasi-free state $\omega$ with matching two-point function and Fock-space representation. Such an algebra can be easily constructed using these Fock space representations.
\begin{definition}
\label{def:extended_algebra}
Let $M$ semi-globally hyperbolic and $\{M_j:j\in J\}$ the ordered cover. For a set $E$ consisting of extended one-particle structures, let $(\mathfrak{H}_{(K,H)},\pi_{(K,H)})$ be the Fock space representation of $\mathcal{A}_{(K,H)}(M)$ associated with $(K,H)\in E$. We now define $\mathcal{A}_E$ the algebra generated by operators $\{\hat\phi_E(f):f\in\mathcal{D}(M)\}$ on $$\mathfrak{H}_E=\bigoplus_{(K,H)\in E}\mathfrak{H}_{(K,H)},$$
defined via
$$\hat\phi_E(f)\left(\Psi_{(K,G)}\right)_{(K,H)\in E}= \left(\pi_{(K,G)}\left(\hat\psi_{(K,G)}(f)\right)\Psi_{(K,G)}\right)_{(K,H)\in E}.$$
\end{definition}
\noindent Clearly, this algebra is a scalar field algebra in the sense of definition \ref{def:scalfield}. It is tempting to view the choice of which one-particle structures one uses to define the theory as some choice of boundary condition. However, while these are not two completely unrelated choices, as we will see for instance in section \ref{sec:punctured}, not every choice of a one-particle structure corresponds to a particular boundary condition. Nevertheless, there is an intimate relationship between the global dynamics of the theory and the set $E$. This is something we make precise below.

\subsection{Dynamics on scalar field algebras}
Recall that quantum field theories on globally hyperbolic space-times ought to satisfy the time-slicing axiom. This can be seen as quantum field theoretical version of theorem \ref{thm:cauchy_prob}, i.e. the fact that there is some well-posed global dynamics. We wish to also have a notion of global dynamics on our scalar field algebras. This is given by the definition below.

\begin{definition}
\label{def:dynamics}
Given a semi-globally hyperbolic space-time with time-ordered cover $M=\cup_{j\in J}M_j$, we say a quantum field theory $\mathcal{A}(M)$ is \textup{past predictive} if for each $j\in J$ we have  $$\mathcal{A}(M;M_j)=\mathcal{A}(M;\cup_{j'\leq j}M_{j'}).$$
Similarly it is \textup{future predictive} if
$$\mathcal{A}(M;M_j)=\mathcal{A}(M;\cup_{j'\geq j}M_{j'}).$$
We say a theory has \textup{complete dynamics} if the theory is both past and future predictive.\footnote{The notion of past predictiveness can also be seen as that a theory does not suffer from information loss, in the sense of \cite{unruhInformationLoss2017}. One could equally say that future predictiveness is the same as that theory does not suffer information gain, however unlike information loss this is not a popular terminology.}
\end{definition}
\noindent Here we should stress once again the choice of time-ordered cover matters for this definition. However here this is very natural, as we have seen in proposition \ref{thm:cover} that the time-ordered cover relates to a choice of semi-Cauchy time-function and hence to a global notion of time. Clearly what notion of time one adheres may play a role in what one would consider the global dynamics. In fact this is already true for globally hyperbolic space-times, where if one chooses a time-function that is not Cauchy (but for instance only semi-Cauchy), the time-slicing axiom fails in a sense. That is to say, once one knows the $n$-point functions of a state in the neighbourhood of an equal time surface, this need not uniquely fix them on the entire space-time.

In the case that our theory has been constructed from a single extended one-particle structure, these predictiveness properties can be one-to-one related with properties of this structure.
\begin{proposition}
Given $M=\cup_{j\in J}M_j$ semi-globally hyperbolic and $(K,H)$ an extended one-particle structure, then $\mathcal{A}_{(K,H)}(M)$ is future predictive if and only if for each $j\in J$
$$K(\mathcal{D}(M_j))=K(\mathcal{D}(\cup_{j'\geq j}M_{j'})).$$
\end{proposition}
\begin{proof}
Clearly if $K(\mathcal{D}(M_j))=K(\mathcal{D}(\cup_{j'\geq j}M_{j'}))$, then for any $f\in\mathcal{\mathcal{D}(\cup_{j'\geq j}M_{j'})}$ there exists an $g\in\mathcal{D}(M_j)$ such that $K(f-g)=0$ and hence
$$\mu(f-g,f-g)=0.$$
Hence $\hat\phi_{(K,H)}(f)=\hat\phi_{(K,H)}(g)$ which implies $$\mathcal{A}_{(K,H)}(M;M_{j})\subset\mathcal{A}_{(K,H)}(M;\cup_{j'\geq j}M_{j'})\subset \mathcal{A}_{(K,H)}(M;M_{j}).$$
\newline
Reversely, given that for each $f\in\mathcal{D}(\cup_{j'\geq j}M_{j'})$ there is a $b\in\mathcal{A}(M;M_j)$ such that
$$b=\hat\phi_{(K,H)}(f).$$
Then there are $g_1,...,g_N\in\mathcal{D}(M_j)$ such that
$$b=\sum_{k_1,...k_N=0}^Mc_{k_1,...,k_N}(\hat\phi(f_1))^{k_1}...(\hat\phi(f_N))^{k_N},$$ and using a similar calculation as the proof of proposition \ref{prop:fock_faith}, we find that
$c_{k_1,...,k_N}=0$ for $k_1+...+k_N>1$, so we know that (given that $b$ is self-adjoint) there is some $c\in\mathbb{R}$ and $g\in\mathcal{D}(M_j)$ such that
$$b=c+\hat\phi_{(K,H)}(g),$$
which implies $\hat\phi_{(K,H)}(f-g)=c$, and from the definition of $\mathcal{A}_{(K,H)}$ this can only be so for $c=0$, which means that $$\mu(f-g,f-g)=0$$ and hence $K(f-g)=0$. Therefore $K(\mathcal{D}(M_j))=K(\mathcal{D}(\cup_{j'\geq j}M_{j'}))$.
\end{proof}

\noindent We have now shown how, given a non-empty set $E$ of extended one-particle structures on semi-globally hyperbolic space-times $M$, a quantum field theory $\mathcal{A}_E(M)$ can be constructed such that each $(K,H)\in E$ corresponds to a one-particle subspace of a Fock-space representation on $\mathcal{A}_E(M)$. Furthermore we have discussed some aspects of dynamics on these theories. In the next section we will consider some example space-times and show how in various contexts a choice of $E$ can lead to theories with varying degrees of predictiveness.

\section{Applications to maximally semi-globally hyperbolic space-times}
\label{sec:maxsem}
Here we consider maximally semi-globally hyperbolic space-times. Unless stated otherwise we always define our scalar field algebras with respect to the maximal cover, which means that these theories behave as expected on each causally convex globally hyperbolic neighbourhood. Firstly, we treat space-times that are globally hyperbolic up to a missing point (and hence have a unique completion to a globally hyperbolic space-time), then we will treat space-times with a larger missing region, such that in particular embeddings into globally hyperbolic space-times are highly non-unique, and thirdly we will treat some space-times that cannot be embedded into globally hyperbolic space-times, with particular attention to black hole evaporation space-times.
\subsection{Punctured globally hyperbolic space-times}
\label{sec:punctured}
The most elementary class of semi-globally hyperbolic space-time, beyond globally hyperbolic space-times themselves, are space-times that are globally hyperbolic up to a missing point. So in particular, for $N$ globally hyperbolic, and some $p\in N$, the space-time $M=N\setminus\{p\}$ is semi-globally hyperbolic. Note in particular that any Cauchy time-function $T:N\rightarrow \mathbb{R}$ restricts to a semi-Cauchy time-function $T\restriction_M$. It is not difficult to see that this space-time is in fact maximally semi-globally hyperbolic, where the maximal cover is given by $M_-=N\setminus \overline{I^+(p)}$ and $M_+=N\setminus \overline{I^-(p)}$, as illustrated in figure \ref{fig:punct}.
\begin{figure}[h]
	\begin{center}
		\includegraphics{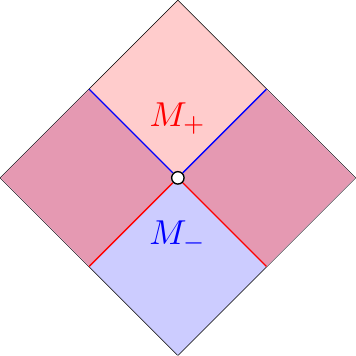}
	\end{center}
	\caption{A punctured globally hyperbolic space-time with maximal cover}
	\label{fig:punct}
\end{figure}
That this space-time admits a quantum field theory should be clear from the fact that it can be embedded in a globally hyperbolic space-time. Each quasi-free state $\omega_N$ on $\mathcal{A}(N)$ defines a one-particle structure $(K_N,H_N)$ on $(N,\sigma_N)$, which can be restricted to $M_{\pm}$ to yield a compatible pair of one-particle structures, and so forth. One may however wonder if this is the only interesting quantum field theory one can construct on such a space-time.

In many cases, one wouldn't expect that the absence of a single point in a space-time should leave an imprint on the overall physics. After all, what kind of physics can be associated with a single point? Certainly in classical field theories, where one is often interested in smooth solutions to equations of motions on smooth space-times, removing a point out of the background will generally not yield any ambiguities of solutions, as long as we restrict ourselves to smooth solutions. In fact, such a result also holds in quantum field theory, where the role of smooth solutions to the equation of motion are now played by (local) Hadamard states, that play an essential role in the definition of higher order observables such as the stress-energy tensor, Wick products and time-ordered products, see \cite{waldQuantumFieldTheory1994, hollandsLocalWickPolynomials2001,hollandsExistenceLocalCovariant2002}. In fact as a consequence of \cite[theorem 3.5]{verchContinuitySymplecticallyAdjoint1997}, we see that given that $\omega,\omega'$ Hadamard states on a scalar field algebra $\mathcal{A}(M)$, by which we mean that their restriction to $\mathcal{A}(M;M_\pm)$ defines a Hadamard state in the sense of \cite{radzikowskiMicrolocalApproachHadamard1996}, it holds that 
$$\omega=\omega'\iff\omega\restriction_{\mathcal{A}(M;M_-)}=\omega'\restriction_{\mathcal{A}(M;M_-)}\iff\omega\restriction_{\mathcal{A}(M;M_+)}=\omega'\restriction_{\mathcal{A}(M;M_+)}.$$

\noindent Nevertheless, is some contexts one may also be interested in states that are not Hadamard. This is especially true when one views the missing point in our space-time as some place-holder for singular behaviour, say a point interaction of some kind.	
In globally hyperbolic space-times it is well known via the propagation of singularities theorem (see \cite{radzikowskiMicrolocalApproachHadamard1996}) that if a state satisfies the Hadamard condition around some Cauchy surface, it satisfies these conditions on the entire space-time. The same cannot be said for the space-time $M$ introduced here, firstly it does not contain any global Cauchy surface, but at a more fundamental level the missing point can be a source for singular behaviour in distributions. That is to say, suppose that a 2-point function $w_2$ on the space-time $M$ satisfies the Hadamard conditions on $M_-$, i.e. following the definitions of \cite{radzikowskiMicrolocalApproachHadamard1996}
$$WF'(w_2\restriction_{D(M_-)^2})=\{(x,k;x',k')\in T(M_-\times M_-)\setminus\mathbf{0}:(x,k)\sim(x',k'), k\vartriangleright0\},$$
then 
\begin{align*}
	WF'(w_2)\subset&\{(x,k;x',k')\in T(M\times M)\setminus\mathbf{0}:(x,k)\sim(x',k'), k\vartriangleright0\}\\
	&\cup\{(x,k;x',k')\in T(M\times M)\setminus\mathbf{0}:TM_+\ni(x,k)\sim(p,\xi)\\
	&\;\;\;\;\;\text{ or }TM_+\ni(x',k')\sim(p,\xi)\text{ for some }\xi\in T_pN\}.
\end{align*}

\noindent Let us give an example. Let $(K_-,H_-)$ and $(K_+,H_+)$ compatible Hadamard one-particle structures inducing faithful Fock space representations on $M_-$ and $M_+$ respectively, then we can define
$$K_+'\rightarrow\mathcal{D}(M_+)\rightarrow H_+\times \mathbb{C},$$
$$K_+'f=(K_+f,G^+(f,p)),$$
$$H_+'=\overline{K_+'(\mathcal{D}(M_+))+iK_+'(\mathcal{D}(M_+))},$$
such that for $f,g\in\mathcal{D}(M_+)$ $$\langle K_+'f,K_+'g\rangle_+'=\langle K_+f,K_+g\rangle_++G^+(f,p)G^+(g,p).\footnote{Note that the new contribution to this inner product has to be real, otherwise it would not respect the pre-symplectic structure on $\mathcal{D}(M_+)$.}$$
By \cite[theorem 3.5]{verchContinuitySymplecticallyAdjoint1997} we can deduce from the Hadamard property that $K_1(D(M_1\cup M_2))+iK_1(D(M_1\cup M_2))$ is dense in $H_1$, therefore 
$$(K,H)=(K_-,H_-)\Vert(K_+,H_+),\;(K',H')=(K_-,H_-)\Vert(K_+',H_+')$$ are the unique extended one-particle structures of these compatible pairs. Since $(K,H)$ corresponds to a quasi-free state on $\mathcal{A}(N)$, we see that
$$\mathcal{A}_{(K,H)}(M)\cong\mathcal{A}(N;M)=\mathcal{A}(N).$$
In particular, $\mathcal{A}_{(K,H)}(M)$ has complete dynamics.
However, in the case of $(K',H')$, the best we can do is find a surjective homomorphism,
$$\mathcal{A}_{(K',H')}(M)\twoheadrightarrow\mathcal{A}(N),$$
with $\hat\phi_{(K',H')}(f)\mapsto\hat\phi_N(f)$. More specifically
$$\mathcal{A}(N)\cong\left.\mathcal{A}_{(K',H')}(M)\middle/ \mathcal{I}\right.,$$
with $\mathcal{I}$ the smallest ideal containing all $\hat\phi_{(K',G')}(f)$ with $$f\in\mathcal{D}(M)\cap(\Box-V)\mathcal{D}(N).$$

It is worth mentioning, that dividing out the ideal $\mathcal{I}$ is not the only way to reduce the theory $\mathcal{A}_{(K',H')}(M)$ to $\mathcal{A}(N)$. 
For example, given $\lambda_1,\lambda_2\in\mathbb{R}$, we can define $\mathcal{I}_{\lambda_1,\lambda_2}$ as the smallest ideal such that for any $f\in\mathcal{D}(M)\cap(\Box-V)\mathcal{D}(N)$ $$\hat\phi_{(K',G')}(f)+\lambda_1G^+(f,p)+\lambda_2G^-(f,p)\in\mathcal{I}_{\lambda_1,\lambda_2}.$$ Even though we have now chosen a different ideal, we see that also in this case we have
$$\left.\mathcal{A}_{(K',H')}(M)\middle/ \mathcal{I}_{\lambda_1,\lambda_2}\right.\cong \mathcal{A}(N).$$
This is not surprising, as one can view the fields $[\hat\phi_{(K',H')}(f)]_{\mathcal{I}_{\lambda_1,\lambda_2}}$ as quantizing the classical equation of motion $$(\Box-V)\phi=-(\lambda_1+\lambda_2)\delta(.,p)$$ on $N$. It is well known that quantizing the real scalar field with any external (classical) source always results in the same algebra (see for instance \cite{fewsterAlgebraicQuantumField2019}).

These examples illustrate that different choices of extended one-particle structures do not necessarily correspond with different choices of `boundary conditions', whatever this means for a discrete boundary. In this case it may be sensible to think of a source term at this point $p$ as corresponding to a boundary condition on $M$. Of course as discussed different choices of these boundary conditions generate the same algebra with the same net-structure, and hence could be though of as equivalent. What we have seen from the theories defined by $(K,H)$ and $(K',H')$ is that they differ in how much `uncertainty' of such a boundary condition is present in the theory, rather than what actually the boundary condition is. This intimitely relates to whether the theory defined has complete dynamics or not.
\subsection{Space-times with macroscopic gaps}
\label{sec:mac_gap}
In the previous subsection we considered space-times that were globally hyperbolic up to one missing point, but of course this can be generalized to larger `missing regions', for instance a space-like disk. Let us give a slightly more abstract description of what kind of space-time we consider in this section.

Let $M=M_-\cup M_+$ maximally semi-globally hyperbolic with $T:M\rightarrow\mathbb{R}$ a semi-Cauchy time-function such that $T^{-1}(\{0\})$ is a smooth Cauchy surface of $M_-\cap M_+$, and $\Sigma_{\pm t}=T^{-1}(\{\pm t\})$ a Cauchy surface of $M_{\pm}$ for each $t>0$. So far this is nothing new, but now we assume that $$\Sigma_{-t}\cong\Sigma_t,$$
and that $M$ can be isometrically embedded in some $N$ such that $$\Sigma_\pm=\lim_{t\downarrow 0}\Sigma_{\pm t},$$
which should be interpreted as the sets of all limits of sequences in $\prod_{n\in\mathbb{N}}\Sigma_{\pm \frac{1}{n}}$, are also smooth Cauchy surfaces of $\mathcal{D}(M_\pm)\subset N$. Furthermore we require that $\Sigma_0=\Sigma_+\cap\Sigma_-$ and that there is a diffeomorphism $$d:\Sigma_-\rightarrow\Sigma_+,$$ that leaves $\Sigma_0$ invariant.

Note that nowhere we have explicitly required that the ambient space-time $N$ is in itself globally hyperbolic, though in practice we can often choose it to be that way. In this case such an embedding can be used to directly define a quantum field theory on $M$ via $\tilde{\mathcal{A}}(M)=\mathcal{A}(N;M)$ which, due to the fact that $N$ is by no means uniquely defined, already shows that there are a lot of different free scalar quantum field theories that can be defined on $M$. However, there is another way one can construct such quantum field theories, which do not make direct use of $N$, but only use the surfaces $\Sigma_\pm$ and the diffeomorphism $d$.

Suppose we have a one-particle structure $(K_-,H)$ on $M_-$. Consider a test-function $f_+\in \mathcal{D}(M_+)$, we know that $\Delta_{M_+}(f_+)\in \mathcal{E}(M_+)$ induces a function $\varphi_+(f_+)=\Delta_{M_+} (f_+)\restriction_{\Sigma_+}\in\mathcal{D}(\Sigma_+)$ and a smooth distribution $\pi_+(f_+)\in\mathcal{E}'(\Sigma_+)$ given by $$(\pi_+(f_+),u)=\int_{\Sigma_+}dA(u\nabla_{\mathbf{n}}\Delta_{M_+}(f_+)).$$
We can now define $\varphi_-(f_+)=\varphi_+(f_+)\circ d\in \mathcal{D}(\Sigma_+)$
and $\pi_-(f_+)\in \mathcal{E}'(\Sigma_-)$ via $$(\pi_-(f_+),u)=(\pi_+(f_+),u\circ d^{-1}).$$
It is not hard to see that $\pi_-(f_+)$ has a smooth compact integration kernel and hence via theorem \ref{thm:cauchy_prob} $\varphi_-(f_+)$ and $\pi_-(f_+)$ define a unique spatially compact solution to the Klein-Gordon equation $\Delta_-(f_+)\in\mathcal{E}(M_-)$. We can now find an $f_+'\in\mathcal{D}(M_-)$ such that $\Delta_{M_-}(f_+')=\Delta_-(f_+)$, which we use to define $K_+:\mathcal{D}(M_+)\rightarrow H$ via $K_+(f_+)=K_-(f_+')$. Here we note that $K_-(f_+')$ is independent of the choice of $f_+'\in\Delta_{M_-}^{-1}(\{\Delta_-(f_+)\})$.
One can show that $(K_+,H)$ defines a one-particle structure on $M_+$. After all for $f_+,g_+\in\mathcal{D}(M_+)$ we have 
\begin{align*}
	2Im(\langle K_+f_+,K_+g_+\rangle)=&2Im(\langle K_-f_+',K_-g_+'\rangle)\\
	=&\sigma_{M_-}(f_+',g_+')\\
	=&(\pi_-(g_+),\varphi_-(f_+))-(\pi_-(g_-),\varphi_-(f_-))\\
	=&(\pi_+(g_+),\varphi_+(f_+))-(\pi_+(g_-),\varphi_+(f_-))\\
	=&\sigma_{M_+}(f_+,g_+),
\end{align*}
furthermore $(K_+\circ(\Box-V) )(f_+)=K_-(0)=0$, and since $d$ is invertible $K_+(\mathcal{D}(M_+))=K_-(\mathcal{D}(M_-))$. From the fact that $d$ leaves $\Sigma_0$ invariant, we also see that $$K_-\restriction_{\mathcal{D}(M_-\cap M_+)}=K_+\restriction_{\mathcal{D}(M_-\cap M_+)},$$ which implies in particular that $(K_-,H)$ and $(K_+,H)$ form a compatible pair.

In principle we can now define $(\tilde{K},H)=(K_-,H)\Vert(K_+,H)$. However in general this one-particle structure does not lead to a theory with well-posed global dynamics. In this particular case, we can define $K:\mathcal{D}(M)\rightarrow H$ more directly, where for $f=f_++f_-$ with $f_\pm\in\mathcal{D}(M_\pm)$ we set
$$Kf=K_+f_++K_-f_-.$$
The quantum field theory $\mathcal{A}_{(K,H)}(M)$ that this extended one-particle structure $$(K,H)$$ defines actually only depends on the choice of $d$ and not on $(K_-,H)$. In fact we can give an alternative construction of this algebra that bypasses the use of these extended one-particle structures. We first extend the maps $\pi_-$ and $\varphi_-$ to $\mathcal{D}(M)$ in the natural way, i.e. such that for $f_-\in\mathcal{D}(M_-)$ we have $\varphi_-(f_-)=\Delta_{M_-}(f_-)\restriction_{\Sigma_-}$ and so forth. We then see that
$$\mathcal{A}_{(K,H)}(M)=\mathcal{A}_d(M)=\left.\mathcal{B}(M)\middle/\mathcal{J}_d\right.,$$
with $\mathcal{J}_d$ the smallest ideal such that for each $f,g\in\mathcal{D}(M)$ we have that $$[\hat\psi(f),\hat\psi(g)]-i((\pi_-(g),\varphi_-(f))-(\pi_-(g),\varphi_-(f)))\in \mathcal{J}_d,$$
and given that $\varphi_-(f)=0$ and $\pi_-(f)=0$ we have 
$$\hat\psi(f)\in\mathcal{J}_d.$$
Note also that this theory has complete dynamics, which is due to the fact that $d$ is a diffeomorphism. Nevertheless, as we will see in section \ref{sec:top_change}, the construction above can be generalized to the case where $d$ is a manifold embedding instead of a diffeomorphism, but in this case these theories will generally not be either past or future predictive.
\subsection{Black hole evaporation space-times}
\label{sec:bhe}
The examples we have considered so far could mostly be embedded in globally hyperbolic space-times. While these examples did allow us to play around with the constructions outlined in section \ref{sec:constr}, these space-times itself, or rather the quantum field theories on these space-times that deviate from the theories on their embedding into globally hyperbolic space-times, are of limited physical interest. We now focus our attention on a class of maximally semi-globally hyperbolic space-time that have played a relevant role in the physics literature over the past decades, namely the `black hole evaporation' space-times. First introduced in \cite{hawkingParticleCreationBlack1975}, (spherically symmetric) black hole evaporation space-times are characterized by their Penrose diagram, drawn in figure \ref{fig:bhe}, and supposedly model a black hole that loses energy due to Hawking radiation (a quantum field theoretical phenomenon which was also first discussed in \cite{hawkingParticleCreationBlack1975}) and are thought to shrink into nothingness due to the semi-classical back-reaction of this radiation on the space-time metric. It should be noted that no explicit semi-classical solution to general relativity coupled to a scalar quantum field theory (via the semi-classical Einstein equations) is known that models formation and evaporation of a black hole, nor is it known if such solutions actually exist. Nevertheless the suggestion that such solutions may exist in this framework has lead to significant controversy, which often goes under the name of the information loss paradox \cite{belotHawkingInformationLoss1999,marolfBlackHoleInformation2017}. Covering this discussion goes beyond the scope of this text. Nevertheless we do reiterate the point of view pushed most notably in \cite{unruhInformationLoss2017} that, as already alluded to when we wrote down definition \ref{def:dynamics}, from the perspective of quantum field theory on curved space-time, information loss in itself is not at odds with any theoretical foundations. Quantum field theories on semi-globally hyperbolic space-times need not have complete dynamics (with respect to some semi-Cauchy time-function). The time-slicing property is only assumed to hold with respect to Cauchy surfaces. Since black hole evaporation space-times are not globally hyperbolic, information loss (or even information gain) is, at the level of semi-classical physics, not at all problematic. Of course when one believes that semi-classical physics approximates an underlying quantum gravity theory that by assumptions does not admit information loss, one would either need to explain how information losing processes can emerge in the semi-classical limit, or argue that black hole evaporation does not occur in the way that one would expect based on semi-classical arguments.

So having established that in principle a theory admitting information loss on black hole evaporation space-times is not against any foundations of quantum field theory, a question that should be addressed is whether or not one can actually construct such a theory in the first place. Using the constructions that we have considered in this paper, we argue that black hole evaporation space-times do admit linear scalar quantum field theories, however we will see that this leads to a problem that is arguably more problematic than information loss (at least from a semi-classical point of view), as it is unclear if the resulting quantum field theory admits any (physically reasonable) states.\\
\newline
Given that $M$ is a black hole evaporation space-time, the construction of $\mathcal{A}(M)$ is based on the observation that $M$ can be approximated by space-times with macroscopic gaps as discussed in section \ref{sec:mac_gap}. That is to say, the maximally globally hyperbolic space-time $M=M_+\cup M_-$ contains a nested sequence $M^{(n)}=M^{(n)}_+\cup M^{(n)}_-$ in the class of space-times introduced in section \ref{sec:mac_gap}, with $$M^{n}_\pm\subset M^{n+1}_\pm\subset M_\pm\subset M,$$
such that 
$$M=\bigcup_nM^{(n)},\text{ and } D(M^{(n)}_\pm)=M_\pm.$$
As we've seen in section \ref{sec:mac_gap} we can always construct an extended one-particle structure on $M^{(n)}$, say $(K^{(n)},H^{(n)})$, which we could even choose such that $\mathcal{A}_{(K^{(n)},H^{(n)})}(M^{(n)})$ has complete dynamics.

Now choose $E^{(n)}$ to be the set of all (unitary inequivalent) extended two-point functions on $M^{(n)}$. Note that there are unique injective maps $\iota_n: E^{(n+1)}\rightarrow E^{(n)}$ such that, given $(K^{(n)},H^{(n)})=\iota_n((K^{(n+1)},H^{(n+1)}))$, for each $f,g\in\mathcal{D}(M^{(n)})\subset \mathcal{D}(M^{(n+1)})$ we have 
$$\langle K^{(n)}f,K^{(n)}g\rangle^{(n)}=\langle K^{(n+1)}f,K^{(n+1)}g\rangle^{(n+1)}.$$
That this map is injective and unique follows from the fact that for each $f\in\mathcal{D}(M^{(n+1)})$ there is (using the wave equation) a function $g\in\mathcal{D}(M^{(n)})$ such that $K^{(n+1)}f=K^{(n+1)}g$. Therefore $K^{(n)}$ uniquely determines $K^{(n+1)}$ (up to a unitary map). By the same argument we have that $H^{(n)}\cong H^{(n+1)}$. This means that there are natural embeddings $$\mathfrak{H}_{E^{(n+1)}}\subset \mathfrak{H}_{E^{(n)}}$$ of the Hilbert spaces as given in definition \ref{def:extended_algebra}, which means in particular that we can define a Hilbert space projector $p^{(n)}:\mathfrak{H}_{E^{(n)}}\rightarrow\mathfrak{H}_{E^{(n+1)}}$.

Shifting our attention to the algebras $\mathcal{A}_{E^{(n)}}(M^{(n)})$, we note that we can define surjective maps
$$\kappa^{(n)}:\mathcal{A}_{E^{(n)}}(M^{(n)})\twoheadrightarrow\mathcal{A}_{E^{(n+1)}}(M^{(n)}),$$
given by $\hat\phi_{E^{(n)}}(f)\mapsto p^{(n)}\hat\phi_{E^{(n)}}(f)p^{(n)}$. That this map is surjective again follows from that the $\hat\phi_{E^{(n)}}$'s satisfy the wave equation. This allows us to define a nested sequence of ideals $\mathcal{I}^{(n)}\subset\mathcal{A}_{E^{(1)}}(M^{(1)})$ defined by 
$$\mathcal{I}^{(n)}=\ker(\kappa^{(n-1)}\circ...\circ\kappa^{(1)}).$$
This yields $$\mathcal{I}^{(n)}\subset \mathcal{I}^{(n+1)},\text{ and } \left.\mathcal{A}_{E^{(1)}}(M^{(1)})\middle/\mathcal{I}^{(n)}\right.\cong\mathcal{A}_{E^{(n)}}(M^{(n)}).$$
We now define
$$\mathcal{A}(M)=\left.\mathcal{A}_{E^{(1)}}(M^{(1)})\middle/\left(\bigcup_n\mathcal{I}^{(n)}\right)\right..$$
where for $f\in\mathcal{D}(M)$ one defines $\hat\phi(f)=[\hat\phi_{E^{(1)}}(f')]_{\bigcup_n\mathcal{I}^{(n)}}$ for $f'\in \mathcal{D}(M^{(1)})$ with $f-f'\in(\Box-V)\mathcal{D}(M)$. This allows $\mathcal{A}(M)$ to be given a net structure in the usual way, where we can see that for each $U\subset M$ causally convex globally hyperbolic
$$\mathcal{A}(M;U)\cong\mathcal{A}(U).$$
From here it is clear that $\mathcal{A}(M)$ is a linear scalar quantum field theory in the sense of definition \ref{def:scalfield}.\\

\noindent Until now all the theories that we constructed admitted states by design, we first constructed a state (or rather one-particle structure) that our theory should admit and built the algebra from there. The construction above does not start from a state on $\mathcal{A}(M)$, but rather from a sequence of what one could arguably call `approximate states', in the sense that these states are defined everywhere on the space-time up to some small neighbourhood of the black hole singularity, which can be made arbitrarily small. In other words, if one allows the physics to deviate (in an a priori unspecified way) from the semi-classical model in an arbitrarily small neighbourhood of the singularity, then one could hope that one of these `approximate states' defines an algebra that describes the physics away from the singularity in a sufficiently accurate way, perhaps even yielding a theory that has complete dynamics. This idea is in itself nothing new, after all one expects that near the singularity, where curvature of the space-time blows up, that the semi-classical approximation breaks down and one needs an underlying theory of quantum gravity to accurately describe the situation. For this reason researchers working in several quantum gravity communities have proposed various quantum gravity mechanisms that could play a role in this regime. A recent example of this is the black hole to white hole tunnelling mechanism that has been studied in \cite{bianchiWhiteHolesRemnants2018}. Such a mechanism would indeed yield a space-time that, outside this transition region, looks exactly like one of these $M^{(n)}$'s, where the precise choice of $(K^{(n)},H^{(n)})\in E^{(n)}$ that can be used to define the quantum field theory on $M^{(n)}$ should somehow follow from the underlying quantum gravity theory.

The question that remains, is whether $\mathcal{A}(M)$ admits any states at all. This is an open question as far as the author is concerned. From the discussion following definition \ref{def:scalfield}, one can construct a Weyl-like algebra $\mathfrak{A}(M)$ associated with $\mathcal{A}(M)$ which does admit states. However whether any of these states can be extended to $\mathcal{A}(M)$, i.e. whether these states define $n$-point functions, should be further analysed. One could hope that there exists some unit vector in $\bigcap_n\mathfrak{H}_{E^{(n)}}$, which if it existed does define a state on $\mathcal{A}(M)$, but unfortunately the set of unit vectors in $\mathfrak{H}_{E^{(n)}}$ is not compact, and only its convex hull is weakly compact (by the Banach-Alaoglu theorem, see \cite{maccluerElementaryFunctionalAnalysis2009}), hence Cantor's intersection theorem does not yield any non-zero elements in $\bigcap_n\mathfrak{H}_{E^{(n)}}$. Therefore, as far as we can tell the only way to decide whether $\mathcal{A}(M)$ admits any physically reasonable states would be an explicit calculation. This calls for further investigation.

\section{Some applications to non-maximally semi-globally hyperbolic space-times}
\label{sec:nonmax}
In the previous section we discussed some constructions of quantum field theories on maximally semi-globally hyperbolic space-times. Due to the nature of these space-times, the theories constructed automatically reduce to the standard theory on globally hyperbolic space-times on each globally hyperbolic causally convex region. However in some cases we would like to construct quantum field theories on more general semi-globally hyperbolic space-times. In this section we give two examples of such a construction.
\subsection{Topology changing space-times}
\label{sec:top_change}
The first example that we will treat is a topology changing space-times. Topology change has mostly been considered in the context of creating wormholes, see for instance \cite{morrisWormholesTimeMachines1988}. As mentioned in this paper, topology change in space-times without naked singularities is only possible when the space-time either contains closed time-like curves or is not time-orientable (see theorem 2 in \cite{gerochTopologyGeneralRelativity1967}). However, our class of semi-globally hyperbolic space-times do allow for naked singularities and hence one may be able to construct worm-hole space-times that do fall in this class. Another example where topology change plays a role are in so-called branching space-times (see for instance \cite{belnapBranchingSpacetime1992}), which are meant to model indeterminism.

While we do not explicitly consider a branching space-time nor a wormhole space-time, here we merely give a proof of principle that quantum field theories can be constructed on at least some topology changing space-time. In particular, we consider a semi-globally hyperbolic space-time $M=M_-\cup M_+$ with $M_+$ having Cauchy surfaces $\Sigma_+\cong \mathbb{R}^3$ and $M_-$ having Cauchy surfaces $\Sigma_-\cong\mathbb{R}^2\times\mathbb{S}$. We assume these space-times have a flat geometry and we have drawn a Penrose diagram of the space-time (having reduced the $\mathbb{R}^2$ symmetry) in figure \ref{fig:top_change}.

\begin{figure}[h]
	\begin{center}
		\includegraphics{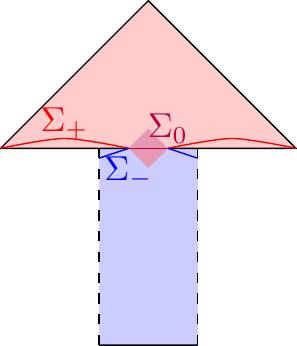}
	\end{center}
	\caption{Penrose diagram of the topology changing space-time $M$ with two spatial directions of translational symmetry suppressed}
	\label{fig:top_change}
\end{figure}

\noindent As one can see in this Penrose diagram, we are nearly in the same situation as in section \ref{sec:mac_gap}, where $\Sigma_0=\Sigma_-\cap\Sigma_+$ is a Cauchy surface of $M_-\cap M_+$, except for the fact that there is no diffeomorphism between $\Sigma_-$ and $\Sigma_+$, but we can find some submanifold $\Sigma_0\subset U\subset\Sigma_-$ such that $d:U\rightarrow\Sigma_+$ is a diffeomorphism leaving $\Sigma_0$ invariant. This allows us to define for each $f_+\in\mathcal{D}(M_+)$ an $\varphi_-(f_+)\in\mathcal{D}(U)\subset\mathcal{D}(\Sigma_-)$ and $\pi_-(f_+)\in\mathcal{E}'(U)\subset\mathcal{E}'(\Sigma_-)$ with smooth integral kernel which we can use to define an extended one-particle structure $(K,H)$ from any one-particle structure $(K_-,H)$ on $M_-$. This allows us to construct a linear scalar quantum field theory on $\mathcal{A}_d(M)$ in exactly the same way as for the macroscopic gap space-times, however with the notable difference that $\mathcal{A}_d(M)$ has no complete dynamics, while it is not past predictive. However it future predictive, as we see that $\mathcal{A}_d(M;M_+)=\mathcal{A}_d(M)$.

\subsection{Approximating time-like boundaries}
As mentioned in the introduction, one class of non-globally hyperbolic space-times on which quantum field theory has already been well-studied in \cite{beniniAlgebraicQuantumField2018}, are (certain) space-times with time-like boundaries. While these space-times are not semi-globally hyperbolic, we can still use our construction to make contact with quantum field theories on these backgrounds.

The class of space-times that we consider here are \textit{globally hyperbolic space-times with time-like boundaries} (as defined in \cite{hauStructureGloballyHyperbolic2019}). These space-times notably still satisfy a version of Geroch's splitting theorem. In particular, let $M$ the bulk of a globally hyperbolic space-time with time-like boundary $\overline{M}$, then there is a time-function $t:M\rightarrow \mathbb{R}$ such that each level surface $t=T$ is an acausal hypersurfaces $\Sigma_T$. In particular $\overline{\Sigma_T}$ is Cauchy in $\overline{M}\cong\mathbb{R}\times \overline{\Sigma_T}$. We claim that we can approximate this space-time by semi-globally hyperbolic space-times. 

We define
$$M^{(n)}=\bigcup_{k=-\infty}^\infty\bigcup_{l=0}^{2^n-1}D\left(\Sigma_{k+\frac{l}{2^n}}\right)\subset M,$$
which clearly yields a nested sequence $$M^{(n)}\subset M^{(n+1)}.$$ Furthermore for each $x\in M$ with $T(x)=t$, we can define $t'=\inf(T^{-1}(J^+(x)\cap\partial M))>t$. Choosing $k,l,m$ such that $t<k+\frac{l}{2^n}<t'$, we see that each inextendible causal curve through $x$ must cross $\Sigma_{k+\frac{l}{2^n}}$ and hence $x\in M^{(n)}$. We conclude that
$$\bigcup_{n\in\mathbb{N}}M_n=M.$$

\noindent Now in principle we could consider $E^{(n)}$ again as the set of all (unitarily inequivalent) extended one-particle structures on $M^{(n)}$. Similarly to the case of the black hole evaporation space-time, this yields a chain of surjective *-isomorphisms
$$\kappa^{(n)}:\mathcal{A}_{E^{(n)}}(M^{(n)})\twoheadrightarrow\mathcal{A}_{E^{(n+1)}}(M^{(n+1)}),$$
and allow us to define
$$\mathcal{A}(M)=\left.\mathcal{A}_{E^{(1)}}(M^{(1)})\middle/\bigcup_{n}\ker(\kappa^{(n)}\circ...\circ\kappa^{(1)})\right..$$
Here the quasi-free states $\mathcal{A}(M)$ can be characterized by the one-particle structures $(K,H)\in E$ on $M$.\footnote{Note that, if non-empty, $E$ can be identified with the inverse limit $\lim_{\leftarrow}E^{(n)}$.}

The construction outlined above is actually inequivalent to the universal extension algebra as constructed in \cite{beniniAlgebraicQuantumField2018}. Firstly $\mathcal{A}(M)$ is constructed such that there exist generators $\{\hat\phi(f):f\in\mathcal{D}(M)\}\in\mathcal{A}(M)$, where $[\hat\phi(f),\hat\phi(g)]$ is in the center of $\mathcal{A}(M)$. This is not true for the universal extension algebra. Another difference is that while the universal extension algebra is constructed to reduce to the standard quantum field algebra on each globally hyperbolic causally convex $U\subset M$, in the case of $\mathcal{A}(M)$ this is only satisfied when $U\subset D(\Sigma_T)$ for some $T\in\mathbb{R}$. Here we see that the fact that definition \ref{def:scalfield} depends on the choice of a semi-Cauchy time-function can actually be a strength instead of a weakness. The construction of \cite{beniniAlgebraicQuantumField2018} as well as the construction above can both be used for implementing local boundary conditions such as the Dirichlet condition (see \cite{dappiaggiCasimirEffectPoint2015}), however if one wants to implement non-local conditions, such as periodic boundary conditions, the former construction can no longer be used. In the case of periodic boundary conditions, where we assume our space-time $M$ to have two diffeomorphic disjoint time-like boundaries $\partial M=B_1\cup B_2$, an identification of these boundaries $\iota:B_1\rightarrow B_2$ induces a new effective causal structure $\leq_\iota$ on the space-time. In general, a causally convex $U\subset M$ need not be causally convex with respect to $\leq_\iota$, which is why once cannot expect a theory respecting these boundary conditions to reduce to the standard theory on $U$ and hence the construction of \cite{beniniAlgebraicQuantumField2018} is not suitable. However, if one can find a time-function $t:M\rightarrow \mathbb{R}$ that is compatible with $\leq_\iota$, i.e. for each $x\leq_\iota x'$ we have $t(x)\leq t(x')$, the construction outlined above is sufficiently versatile to be compatible with these boundary conditions.

However, both constructions generally do not have complete dynamics. Of course in our construction we did not make reference to a particular choice of boundary conditions, so first these will need to actually be imposed by dividing out a further ideal from $\mathcal{A}(M)$. As in \cite{beniniAlgebraicQuantumField2018}, we can construct such an ideal $\mathcal{I}_{G_\pm}$ from an adjoint pair of retarded and advanced propagators $(G_+,G_-)$ satisfying the desired boundary conditions. Here it should be noted that for this algebra $\mathcal{A}_{G_\pm}(M)=\left.\mathcal{A}(M)\middle/\mathcal{I}_{G_\pm}\right.$ to admit any states, we need to make sure that there exists an extended one-particle structure $(K,H)\in E$ such that
$$\int M dV f(G_+-G_-)g=2Im\left(\langle Kf,Kg\rangle\right).$$
As was also concluded in \cite{beniniAlgebraicQuantumField2018}, if one compares the resulting algebra to explicit constructions of quantum field theories, such as made in \cite{dappiaggiCasimirEffectPoint2015}, one finds that the resulting algebra $\mathcal{A}_{G_\pm}(M)$ does not capture all relevant degrees of freedom, only those that can be localized in the interior. This has as a consequence that the dynamics on $\mathcal{A}_{G_\pm}(M)$ is still not complete. In general, one needs to also consider boundary degrees of freedom to construct a complete theory on space-times with time-like boundaries.
\section{Concluding remarks}
We have seen that one can construct linear scalar quantum field theories on various semi-globally hyperbolic space-times. Using the constructions of section \ref{sec:constr} we saw that on any semi-globally hyperbolic space-time we can define a scalar field algebra that admits states if and only if one could define a compatible chain of one-particle structures on that space-time. Nevertheless we have also seen that in some cases, such as in the case of the evaporating black hole space-time, one can construct scalar field algebras, but it is a priori unclear if such an algebra admits any state. On the one hand this is very much in the spirit of algebraic quantum field theory, where one aims to disentangle purely algebraic considerations and questions concerning states and representations. Arguably, the observation that on non-globally hyperbolic space-times, such as the black hole evaporation space-time, constructing an algebra is not so problematic, but that the problems mainly arise when one wants to find physically reasonable states, fits very well in this tradition. On the other hand, the methods that we used to arrive at this construction explicitly made use of one-particle structures and representations of the to be constructed (sub)algebras. This means that, unlike for the well-established construction of linear scalar quantum field theory on globally hyperbolic space-times, the separation between algebras and representations has not been made very cleanly in this work. Whether this is a shortcoming of our constructions, or rather hints at the fact that the conceptual separation of algebras and their representation is a luxury that one cannot always afford, is, as far as the author of this paper is concerned, open for debate.
	
	\section{Acknowledgements}
Firstly I'd like to thank Prof. Dr. Klaas Landsman, under whose supervision I wrote my Master thesis \cite{janssendaanQuantumFieldsNonGlobally2019} that became the foundations of this paper. Secondly I want to thank my doctoral advisor, Prof. Dr. Rainer Verch, under whose supervision I was able to expand on this thesis and write this paper. I also want to thank the Deutsche Forschungsgemeinschaft for their financial support of my doctoral project as part of the Research Training Group 2522: ``Dynamics and Criticality in Quantum and Gravitational Systems" (\url{http://www.rtg2522.uni-jena.de/}).
	\newpage
	\begin{appendices}
\section{Recalling a construction on globally hyperbolic space-times}
\label{sec:qftapp}
In order to construct quantum field theories on semi-globally hyperbolic space-times, we need to recall some facts about the linear scalar (quantum) field on globally hyperbolic space-times. On these space-time \eqref{eq:KG} has a well posed initial value problem, as proven for instance in \cite{barWaveEquationsLorentzian2007}.
\begin{theorem}
	\label{thm:cauchy_prob}
	Let $M$ globally hyperbolic, $\Sigma$ a smooth acausal Cauchy surface, $\mathbf{n}\in \mathfrak{X}(M)$ normal to $\Sigma$, $f\in\mathcal{D}(M)$ and $\varphi,\pi\in\mathcal{D}(\Sigma)$. Then there is a unique $F\in\mathcal{E}(M)$ such that 
	$$(\Box-V)F=f,\;F\restriction_\Sigma=\varphi,\;\nabla_\mathbf{n}F\restriction_\Sigma=\pi.$$
	Furthermore $F\in J(\supp(f)\cup\supp(\varphi)\cup\supp(\pi))$.\footnote{To clarify some (standard) notation, $\mathfrak{X}(M)$ is the set of real vectorfields on $M$, i.e. smooth sections of the tangent bundle $TM$, $\mathcal{D}(M)$ is the vectorspace of smooth real scalar functions on $M$ with compact support (usually referred to as test functions and given an appropriate topology which we do not consider here), $\mathcal{E}(M)$ is the space of arbitrary real smooth functions on $M$ and $J(U)$ is the set of all points in $M$ connected to some point $p\in U\subset M$ via a causal curve.}
\end{theorem}
\noindent It is this (classical) result that motivates the restriction to globally hyperbolic space-times when constructing a quantum field theory. In part, this is related to the following consequence of the theorem above.
\begin{corollary}
	Let $M$ globally hyperbolic, there are two linear maps $$G_\pm:\mathcal{D}(M)\rightarrow \mathcal{E}(M),$$
	referred to as the \textup{advanced} and \textup{retarded propagator}, uniquely defined by 
	$$(\Box-V)G_\pm(f)=f,\; \supp(G_\pm(f))\subset J^\pm(\sup(f)).$$
\end{corollary}
\noindent Using these propagators, one can make the following definition:
\begin{definition}
	The \textup{causal propagator} $$\Delta:\mathcal{D}(M)\rightarrow\mathcal{E}(M),$$ is defined as $$\Delta(f)=G_+(f)-G_-(f).$$
\end{definition}
\noindent This causal propagator plays a key role in the construction of the linear scalar quantum field, in particular in fixing the commutator relations of the algebra of field observables. It has the following properties, as proven in \cite{waldQuantumFieldTheory1994}.
\begin{theorem}
	\label{thm:causprop}
	Let $M$ globally hyperbolic, $\Sigma$ a smooth Cauchy surface, $\mathbf{n}\in\mathfrak{X}(M)$ normal to $\Sigma$. Then
	\begin{enumerate}
		\item The range of the causal propagator $\mathcal{S}_{s.c.}(M):=\Delta(\mathcal{D}(M))$ is the set of all spatially compact solutions to the equation of motion, i.e. $F\in  \mathcal{S}_{s.c}(M)$ iff 
		$$(\Box-V)F=0,\;F\restriction_\Sigma,\nabla_\mathbf{n}F\restriction_\Sigma\in\mathcal{D}(\Sigma).$$
		\item The kernel of $\Delta$ is the linear subspace $(\Box-V)\mathcal{D}(M)\subset \mathcal{D}(M)$.\footnote{Combining point 1 and 2, we see that $\Delta$ defines a linear isomorphism between $\frac{\mathcal{D}(M)}{(\Box-V)\mathcal{D}(M)}$ and $\mathcal{S}_{s.c}(M)$.}
		\item For $f\in\mathcal{D}(M)$ and $F\in\mathcal{E}(M)$ with $(\Box-V)F=0$, we have
		$$\int_M dV fF=\int_\Sigma dA(\Delta(f)\nabla_\mathbf{n}F-F\nabla_\mathbf{n}\Delta(f))\restriction_\Sigma.$$
		\item Defining for $f,g\in\mathcal{D}(M)$ $$\sigma_M(f,g)=\int_M dV f\Delta(g),$$ we have 
		$$\sigma_M(f,g)=\int_\Sigma dA(\Delta(f)\nabla_\mathbf{n}\Delta(g)-\Delta(g)\nabla_\mathbf{n}\Delta(f))\restriction_\Sigma.$$
		As a consequence $\sigma_M$ defines a symplectic form on $\frac{\mathcal{D}(M)}{(\Box-V)\mathcal{D}(M)}$, i.e. a non-degenerate anti-symmetric bi-linear map.
	\end{enumerate}
\end{theorem}
\noindent The symplectic space of classical solutions given above forms the basis for the construction of the (real) scalar quantum field. This construction, which generalizes the more familiar canonical quantization of the scalar field on Minkowski space, can be traced back to \cite{dimockAlgebrasLocalObservables1980}, albeit in a Weyl algebra formulation. Here we will give a formulation more similar to the construction as in \cite[Chapter 3]{brunettiAdvancesAlgebraicQuantum2015}.

Given a certain potential, the algebra of linear observables on $M$ is defined as the so-called CCR algebra of the symplectic space $$\left(\frac{\mathcal{D}(M)}{(\Box-V)\mathcal{D}(M)},\sigma_M\right).$$ Informally, this is given as the unital *-algebra $\mathcal{A}(M)$ generated by elements 
$$\{\hat{\phi}(f):f\in\mathcal{D}(M)\},$$ subject to the following relations:
\begin{align*}
	\forall f,g\in \mathcal{D}(M), a\in\mathbb{R}:\hat{\phi}(af+g)=a\hat{\phi}(f)+\hat{\phi}(g),\\
	\forall f\in \mathcal{D}(M):\hat{\phi}(f)^*=\hat{\phi}\left(f\right),\\
	\forall f,g\in \mathcal{D}(M):[\hat{\phi}(f),\hat{\phi}(g)]=i\sigma_M(f,g)\mathds{1},\\
	\forall f\in \mathcal{D}(M):\hat{\phi}((\Box-V)f)=0.
\end{align*}

\noindent This is made precise in the following definition, which for future purposes constructs $\mathcal{A}(M)$ in two steps.
\begin{definition}
	\label{def:linscal}
	For a (not necessarily globally hyperbolic) space-time $M$, we define the \textup{pre-field algebra of real linear scalar observables} $\mathcal{B}(M)$ in the following way. Let $\langle\mathcal{D}(M)\rangle$ the unital free *-algebra over $\mathbb{C}$ generated by $\mathcal{D}(M)$.\footnote{The unital free *-algebra can be seen as the algebra of polynomials over $\mathbb{C}$ with the noncommuting variables $\langle f\rangle$ and $\langle f\rangle^*$ for $f\in\mathcal{D}(M)$.} Define $\mathcal{I}\subset\langle\mathcal{D}(M)\rangle$ as the smallest ideal such that for each $f,g\in\mathcal{D}(M)$ and $a\in\mathbb{C}$ we have
	$$\langle af+g\rangle-a\langle f\rangle-\langle g\rangle\in\mathcal{I},$$
	$$\langle f\rangle-\langle f\rangle^*\in\mathcal{I}.$$
	Define $$\mathcal{B}(M):=\left.\langle\mathcal{D}(M)\rangle\middle/\mathcal{I}\right.,$$
	and denote for $f\in\mathcal{D}(M)$ $$\hat\psi(f):=[\langle f\rangle]_{\mathcal{I}}.$$
	
	For $M$ globally hyperbolic, we can define the \textup{linear scalar quantum field algebra} $\mathcal{A}(M)$. Let $\mathcal{J}\subset \mathcal{B}(M)$ be the smallest ideal such that for each $f,g\in\mathcal{D}(M)$
	$$[\hat{\psi}(f),\hat{\psi}(g)]-i\sigma_M(f,g)\mathds{1}\in \mathcal{J},$$
	and for $f\in\mathcal{D}(M)$ such that for all $g\in\mathcal{D}(M)$
	$$\sigma_M(f,g)=0,$$
	we also have
	$$\psi(f)\in \mathcal{J}.$$
	We now set
	$$\mathcal{A}(M):=\left.\mathcal{B}(M)\middle/\mathcal{J}\right..$$
	For each $f\in\mathcal{D}(M)$ denote $$\hat\phi(f):=[\hat\psi(f)]_{\mathcal{J}}\in\mathcal{A}(M).$$
	
	Both $\mathcal{B}(M)$ and (for $M$ globally hyperbolic) $\mathcal{A}(M)$ are given a netstructure where for $U\subset M$ open, $\mathcal{B}(M,U)\subset\mathcal{B}(M)$ the smallest unital *-subalgebra such that for each $f\in\mathcal{D}(M)$ with $\text{supp}(f)\subset U$ we have $\psi(f)\in\mathcal{B}(M;U)$, from which can we define for $M$ globally hyperbolic
	$$\mathcal{A}(M;U)=[\mathcal{B}(M;U)]_\mathcal{J}.$$
\end{definition}	
\noindent Observe that the pre-field algebra contains no information on dynamics whatsoever. Both the dynamics of the field and commutation relations are implemented by the quotient $\mathcal{A}(M):=\left.\mathcal{B}(M)\middle/\mathcal{J}\right..$, where we should note that if $\sigma_M(f,g)=0$ for all $g\in\mathcal{D}(M)$, this is equivalent to $\Delta f=0$ and hence also to $f\in(\Box-V)\mathcal{D}(M)$, as we saw in theorem \ref{thm:causprop}.

While there are many equivalent ways of defining this algebra of observables, just as there are many ways of describing effectively the same symplectic space of classical solutions, the construction above lends itself very well to the interpretation of the linear observables as being generated by an `operator valued distribution' with integral kernel $\hat\phi(x)$ satisfying the equation of motion.\\

\noindent It should be noted that the net structure $\mathcal{A}(M;\cdot)$ (i.e. a net of *-algebra's indexed by open subsets $U\subset M$ satisfying $\mathcal{A}(M;U)\subset \mathcal{A}(M;V)$ whenever $U\subset V$), has the following convenient properties, as shown to hold in \cite[Chapter 3]{brunettiAdvancesAlgebraicQuantum2015}.
\begin{theorem}
	\label{thm:net-struc}
	Let $M$ globally hyperbolic, $U,V\subset M$. Then 
	\begin{enumerate}
		\item If $U$ and $V$ not causally related, i.e. there is no causal curve connecting $U$ and $V$, then $$[\mathcal{A}(M;U),\mathcal{A}(M;V)]=\{0\}.$$ This property is referred to as \textup{Einstein causality}.
		\item If there is a Cauchy surface $\Sigma\subset M$ such that $\Sigma\subset U$, then $$\mathcal{A}(M;U)=\mathcal{A}(M).$$ This is known as the \textup{time-slicing property}.
		\item If $U$ causally convex (and as a result, globally hyperbolic in itself, see \cite[prop. 6.6.2]{hawkingLargeScaleStructure1973}), then there is a *-isomorphism 
		$$\iota:\mathcal{A}(U)\rightarrow\mathcal{A}(M;U),$$
		where $\iota(\hat\phi_U(f))=\hat\phi_M(f)$.
	\end{enumerate}
\end{theorem}
\noindent Historically, these net structures of *-algebras were used as the basis for an axiomatic framework to formalize the notion of quantum field theories on Minkowski space-time, see in particular the Haag-Kastler nets as introduced in \cite{haagAlgebraicApproachQuantum1964}. It was this algebraic approach to quantum field theory that formed the basis for Dimock's generalization of the linear scalar field to arbitrary globally hyperbolic space-times in \cite{dimockAlgebrasLocalObservables1980}. In \cite{fewsterAlgebraicQuantumField2015} and references therein it was then recognized that the mathematical structures of nets of algebras on globally hyperbolic space-times could be further generalized into the framework of category theory. Here it is noted that, as a result of point 3 of theorem \ref{thm:net-struc}, $\mathcal{A}$ can be seen as a functor from the category of globally hyperbolic space-times (referred to as $\cat{Loc}$), where morphisms are given by causally convex embeddings, to the category of *-algebras $\cat{Alg}$, where the morphisms are given by *-homomorphisms. This forms the basis of a general axiomatic framework to describe covariant quantum field theories on arbitrary globally hyperbolic space-times.\\

\noindent Of course where there is a *-algebra associated with a physical theory there is a notion of states. We recall the definition.
\begin{definition}
	Let $\mathcal{A}$ be a unital *-algebra. A \textup{state} is a linear map 
	$$\omega:\mathcal{A}\rightarrow \mathbb{C},$$
	such that $\omega(1)=1$ and for each $a\in \mathcal{A}$ we have $\omega(aa^*)\geq0$.
\end{definition}
\noindent For a globally hyperbolic space-time $M$, states on the algebra $\mathcal{A}(M)$ are uniquely characterized by their $n$-point functions 
$$w_n:\mathcal{D}(M)^n\rightarrow\mathbb{C},$$ where $w_n(f_1,...,f_n)=\omega(\hat\phi(f_1)...\hat\phi(f_n))$. We recall an important class of states for the linear scalar fields.
\begin{definition}
	A state $\omega$ on $\mathcal{A}(M)$ is called quasi-free if there is a positive semi-definite symmetric bilinear form (i.e. a real pre-inner product) $\mu:\mathcal{D}(M)^2\rightarrow\mathbb{R}$ satisfying
	$$\vert\sigma_M(f,g)\vert^2\leq 4\mu(f,f)\mu(g,g),$$ such that for each $f\in\mathcal{D}(M)$ we have
	$$\omega\left(\exp(i\hat\phi(f))\right)=\exp(-\frac{1}{2}\mu(f,f)),$$
	where both sides of the equation should be interpreted as a power series expansion and the equality should hold for each order of $f$.
\end{definition}
\noindent It should be noted that via polarization formulas all $n$-point functions of a quasi-free state are uniquely defined by the relation above. In particular 
$$w_2(f,g)=\mu(f,g)+\frac{i}{2}\sigma_M(f,g).$$
Each quasi-free state can be uniquely associated with a one-particle structure.
\begin{definition}
	Given the presymplectic space $(\mathcal{D}(M),\sigma_M)$ a \textup{one-particle structure} $(K,H)$ is given by a Hilbert space $H$ and a real linear map $K:\mathcal{D}(M)\rightarrow H$ such that $K(\mathcal{D}(M))+iK(\mathcal{D}(M))$ is dense in $H$, $K\circ(\Box-V)=0$ and $$2Im(\langle Kf,Kg\rangle)=\sigma_M(f,g).$$
\end{definition}
\noindent Clearly each one-particle structure defines a quasi-free state via $$\mu(f,g)=Re(\langle Kf,Kg\rangle),$$ and hence $w_2(f,g)=\langle Kf,Kg\rangle$. That each quasi-free state defines a unique one-particle structure, is proven in \cite[Appendix A]{kayTheoremsUniquenessThermal1991}.\\

\noindent For a general state $\omega$ on an algebra $\mathcal{A}$ one can find a unique representation (up to unitary equivalence) $(\mathfrak{H}_\omega,\mathfrak{D}_\omega,\pi_\omega,\Omega_\omega)$ such that $\mathfrak{H}_\omega$ a Hilbert space, $\mathfrak{D}_\omega\subset \mathfrak{H}_\omega$ a dense subset, $\pi_\omega:\mathcal{A}\rightarrow L\left(\mathfrak{D}_\omega\right)$ a *-homomorphism and $\Omega_\omega\in\mathfrak{D}_\omega$ a unit-vector with $\pi_\omega(\mathcal{A})\Omega_\omega$ dense in $\mathfrak{H}_\omega$ such that
$$\omega(a)=\langle \Omega,\pi_\omega(a)\Omega\rangle.$$This representation can be obtained using the GNS construction, see for instance \cite[Chapter 5]{brunettiAdvancesAlgebraicQuantum2015}. For quasi-free states on our algebra $\mathcal{A}(M)$ these representations take a special form.
\begin{definition}
	\label{def:fock}
	For $M$ globally hyperbolic, a \textup{(bosonic) Fock-space representation} $(\mathfrak{H},\pi)$ of $\mathcal{A}(M)$ is given by a Hilbert space of the form
	$$\mathfrak{H}=\bigoplus_{n=0}^\infty S_n\left(H^{\otimes n}\right),$$ where $(K,H)$ is some one-particle structure on $(\mathcal{D}(M),\sigma_M)$, $H^{\otimes 0}=\mathbb{C}$ and $S_n:H^{\otimes n}\rightarrow H^{\otimes n}$ the symmetrization operator, where the inner product $\langle.,.\rangle_n$ on $S_n\left(H^{\otimes n}\right)$ is given by
	$$\langle a,b\rangle_0=\overline{a}b,$$
	and for $n>0$
	\begin{multline*}\left\langle \sum_{i_1,...,i_n}a_{i_1,...,i_n}\psi_{i_1}...\psi_{i_n},\sum_{j_1,...,j_n}b_{j_1,...,j_n}\phi_{j_1}...\phi_{j_n}\right\rangle_n=\\\sum_{i_1,j_1...i_n,j_n}\overline{a_{i_1,...,i_n}}b_{j_1,...,j_n}\langle \psi_{i_1},\phi_{j_1}\rangle...\langle \psi_{i_n},\phi_{j_n}\rangle,
	\end{multline*}
	and a *-homomorphism $\pi$ mapping $\mathcal{A}(M)$ into the unbounded operators on $\mathfrak{H}$ such that
	$$\pi(\hat\phi(f))=\hat a(Kf)+\hat a(Kf)^*,$$
	with
	$$\hat a(f)^*(\Psi_0,...,\Psi_n,0,...)=(0,S_1(Kf\otimes \Psi_0),...,S_{n+1}(Kf\otimes \Psi_n),0,...).$$
\end{definition}
\noindent We see that the creation and annihilation operators, $a(\psi)^*$ and $a(\psi)$ respectively (for some $\psi\in H$) satisfy the canonical commutation relations
$$[a(\psi),a(\phi)]=0,\;[a(\psi)^*,a(\phi)^*]=0,\;[a(\psi),a(\phi)^*]=\langle \psi,\phi\rangle.$$ If one constructs the Fock space as a representation of the Weyl algebra, one sees immediately that this representation is faithful from the fact that Weyl algebras are simple (see \cite[Theorem 5.2.8]{bratteliStatesQuantumStatistical1981}) However for completeness we give a direct proof for Fock space representations of $\mathcal{A}(M)$.
\begin{proposition}
	\label{prop:fock_faith}
	The Fock space representations of a linear scalar field theory $\mathcal{A}(M)$ for $M$ globally hyperbolic are faithful, i.e. for each $b\in\mathcal{A}(M)$, $$\pi(b)=0\iff b=0.$$
\end{proposition}
\begin{proof}
	First note that for $f\in\mathcal{D}(M)$ $Kf=0$ if and only if $\hat\phi(f)=0$. This follows from the fact that $Kf=0$ implies $\sigma(f,g)=0$ for all $g\in\mathcal{D}(M)$ which means $f\in(\Box-V)\mathcal{D}(M)$.
	
	Now suppose $b\in\mathcal{A}(M)$. Without loss of generality there are a finite number $f_j\in\mathcal{D}(M)$ for $j=1,...,N$ such that $\psi_j=Kf_j\in H$ linearly independent and an $M\in\mathbb{N}$ such that
	$$b=\sum_{k_1,...k_N=0}^Mc_{k_1,...,k_N}(\hat\phi(f_1))^{k_1}...(\hat\phi(f_N))^{k_N},$$
	which means
	$$\pi(b)=\sum_{k_1,...k_N=0}^Nc_{k_1,...,k_N,l_1,...,l_N}(a(\psi_1)^*+a(\psi_1))^{k_1}...(a(\psi_N)^*+a(\psi_N))^{k_N}.$$
	Note that $H_N=\text{span}(\psi_1,...,\psi_N)$ is an $N$-dimensional Hilbert space, on which we can find a basis $\{\phi_j\in H_N:j=1,...,N\}$ such that $$\langle \psi_j,\phi_j\rangle=\delta_{ij}.$$
	Assume that $\pi(b)=0$. This must in particular mean that $[\pi(b),A]=0$ for any linear operator $A$ on $\mathfrak{H}$. 
	Defining $$D_A(B)=[A,B],$$ we can now calculate
	$$c_{M,...,M}=(D_{a(\phi_1)}^M\circ...\circ D_{a(\phi_N)}^M)(\pi(b))=0.$$
	This can be repeated for all lower order contributions to yield $c_{k_1,...,k_N}=0$. Thus $b=0$.
\end{proof}

\noindent A key conceptual result (see for instance \cite[Chapter 5]{brunettiAdvancesAlgebraicQuantum2015}) is that the class of Fock-space representations (up to unitary equivalence) exactly matches the class of GNS representations for a quasi-free state $\omega$ per the correspondence of quasi-free states to one-particle structures. Here of course $\mathfrak{H}_\omega=\mathfrak{H}$ and $\pi_\omega=\pi$, but furthermore $\Omega_\omega=(1,0,....)$ and $\mathfrak{D}_\omega=\pi(\mathcal{A}(M))\Omega_\omega$. These representations allow for a particle interpretation where $\hat a(\psi)$ can be interpreted as creating a particle in the one particle state $\psi\in H$ (taken to be a unit vector) such that we can define number operators $N(\psi)=\hat a^*(\psi)\hat a(\psi)$, whose eigenvectors can be interpreted as corresponding to states where a finite number of particles in the one-particle state $\psi$ have been excited from the vacuum state $\Omega_\omega$. The eigenvalue then corresponds to the number of these particles. In the absence of sufficient symmetries there is generally no clear way to select a preferred particle interpretation and associated vacuum state, after all for free scalar fields of Minkowski space-time a preferred Fock-space representation is selected by the fact that it allows a unitary implementation of space-time symmetries that leave particle numbers invariant. Nevertheless we still see the fact that a theory allows for Fock-spaces and particle interpretations in the first place as a very useful, if not necessary feature of a (linear scalar) quantum field theory. Therefore in section \ref{sec:qft}, where we discuss the construction of linear scalar quantum fields on semi-globally hyperbolic space-times, we demand that these theories allow for Fock-space representations and in fact use the one-particle structures as a main ingredient for the construction of our theories.
	\end{appendices}
\printbibliography
\end{document}